\documentclass[11pt,a4paper]{article}
\usepackage[a4paper]{geometry}
\usepackage{amssymb,latexsym,amsmath,amsfonts,amsthm}
\usepackage{graphicx}
\usepackage{epstopdf}
\usepackage{float,amsmath,amssymb,mathrsfs,bm,multirow,graphics}
\usepackage{tikz}
\usepackage{pgflibraryshapes}
\usetikzlibrary{arrows,decorations.markings}
\usepackage{subfigure}
\usepackage{overpic}
\usepackage{fullpage}
\usepackage{comment,verbatim}
\usepackage{hyperref}
\usepackage{color}
\usepackage{mathrsfs}
\usepackage[title,titletoc]{appendix}
\usepackage{ulem}
\usepackage{enumitem}

\newcommand{\Bes}{\mathrm{Bes}}

\renewcommand{\Re}{\mathrm{Re}\,}
\renewcommand{\Im}{\mathrm{Im}\,}
\newcommand{\Ai}{\mathrm{Ai}}

\newcommand{\ud}{\,\mathrm{d}}

\newcommand{\what}{\widehat}

\newcommand{\ii}{\mathrm{i}}

\newcommand{\Boh}{\mathcal{O}}

\newtheorem{theorem}{Theorem}[section]
\newtheorem{lemma}[theorem]{Lemma}
\newtheorem{proposition}[theorem]{Proposition}

\newtheorem{corollary}[theorem]{Corollary}

\newtheorem{rhp}[theorem]{RH problem}

\theoremstyle{definition}

\theoremstyle{remark}

\newtheorem{remark}[theorem]{Remark}

\newtheorem*{notation}{Notations}
\numberwithin{equation}{section}

\begin{document}
\title{Asymptotic expansion of the hard-to-soft edge transition}
\author{Luming Yao\footnotemark[1] ~~and ~Lun Zhang\footnotemark[2]}

\renewcommand{\thefootnote}{\fnsymbol{footnote}}
\footnotetext[1]{Institute for Advanced Study, Shenzhen University, Shenzhen 518060, China. E-mail: lumingyao\symbol{'100}szu.edu.cn.}
\footnotetext[2]{School of Mathematical Sciences, Center for Applied Mathematics and Shanghai Key Laboratory for Contemporary Applied Mathematics, Fudan University, Shanghai 200433, China. E-mail: lunzhang\symbol{'100}fudan.edu.cn.}
\date{\today}
\maketitle
\begin{abstract}
By showing that the symmetrically transformed Bessel kernel admits a full asymptotic expansion for large parameter, we establish a hard-to-soft edge transition expansion. This resolves a conjecture recently proposed by Bornemann.
\end{abstract}


\section{Introduction and statement of results}
Consideration in this paper is a universal phenomenon arising from the random matrix theory, namely, the hard-to-soft edge transition \cite{BF03}. As a concrete example, let $X_1$ and $X_2$ be two $n \times (n+\nu)$, $\nu \geq 0$, random matrices, whose element is chosen to be an independent normal random variable. The $n \times n$ complex Wishart matrix $X$ or the Laguerre unitary ensemble (LUE), which plays an important role in statistics and signal processing (cf. \cite{IEEE,Wishart} and the references therein), is defined to be
$$
X=(X_1+\ii X_2)(X_1+\ii X_2)^\ast,
$$
where the superscript $^\ast$ stands for the operation of conjugate transpose. As $n\to \infty$ with $\nu$ fixed, the smallest eigenvalue of $X$ accumulates near the hard-edge $0$. After proper scaling, the limiting process is a determinantal point process characterized by the Bessel kernel \cite{Forr93,FN}
\begin{equation}\label{def:Beskernel}
K_{\nu}^{\textrm{Bes}}(x,y):=\frac{J_{\nu}(\sqrt{x})\sqrt{y}J_{\nu}'(\sqrt{y})-J_{\nu}(\sqrt{y})\sqrt{x}J_{\nu}'(\sqrt{x})}{2(x-y)}, \qquad x,y>0,
\end{equation}
where $J_{\nu}$ is the Bessel function of the first kind of order $\nu$ (cf. \cite{DLMF}). If the parameter $\nu$ grows simultaneously with $n$ in such a way that $\nu/n$ approaches a positive constant, it comes out that the smallest eigenvalue is pushed away from the origin, creating a soft-edge. The fluctuation around the soft-edge, however, is given by the Airy point process \cite{Forr93,FN}, which is determined by the Airy kernel
\begin{align}\label{def:KAi}
K^{\Ai}(x,y) = \frac{\Ai (x) \Ai'(y)-\Ai'(x)\Ai(y)}{x-y}, \qquad   x,y\in\mathbb{R},
\end{align}
where $\Ai$ is the standard Airy function. One encounters the same limiting process by considering the scaled cumulative distribution of largest eigenvalues for large random Hermitian matrices with complex Gaussian entries, which is also known as the Tracy-Widom distribution \cite{TW94}.

Besides the above explanation of the hard-to-soft edge transition, the present work is also highly motivated by its connection with distribution of the length of longest increasing subsequences. Let $S_n$ be the set of all permutations on $\{1,2,\ldots,n\}$. Given $\sigma\in S_n$, we denote by $L_n(\sigma)$ the length of longest increasing subsequences, which is defined as the maximum of all $k$ such that $1\leq i_1<i_2<\cdots<i_k\leq n$ with $\sigma(i_1)<\sigma(i_2)<\cdots<\sigma(i_k)$. Equipped $S_n$ with the uniform measure, the question of the distribution of discrete random variable $L_n(\sigma)$ for large $n$ was posed by Ulam in the early 1960s \cite{Ulam}. After the efforts of many people (cf. \cite{AD99,Romik} and the references therein), Baik, Deift and Johansson finally answered this question in a celebrated work \cite{BDJ} by showing
\begin{equation}\label{eq:LnF}
\lim_{n\to \infty}\mathbb{P}\left(\frac{L_n-2\sqrt{n}}{n^{1/6}}\leq t\right) = F(t),
\end{equation}
where
\begin{equation}\label{def:F}
F(t):=\det(I-K^{\Ai})\big|_{L^2(t, \infty)}
\end{equation}
is the aforementioned Tracy-Widom distribution with $K^{\Ai}$ being the Airy kernel in \eqref{def:KAi}.

To establish \eqref{eq:LnF}, a key ingredient of the proof is to introduce the exponential generating function of $L_n$ defined by
$$
P(r;l)=e^{-r}\sum_{n=0}^{\infty}\mathbb{P}(L_n \leq l)\frac{r^n}{n!}, \qquad r>0,
$$
which is known as Hammersley's Poissonization of the random variable $L_n$. The quantity itself can be interpreted as the cumulative distribution of $L(r)$ -- the length of longest up/right path from $(0,0)$ to $(1,1)$ with nodes chosen randomly according to a Poisson process of intensity $r$; cf. \cite[Chapter 2]{BDS}. By representing $P(r;l)$ as a Toeplitz determinant, it was proved in \cite{BDJ} that
\begin{equation}\label{eq:LF}
\lim_{r \to \infty}\mathbb{P}\left(\frac{L(r)-2\sqrt{r}}{r^{1/6}}\leq t\right) = F(t).
\end{equation}
This, together with Johansson's  de-Poissonization lemma \cite{Johan98}, will lead to \eqref{eq:LnF}.

Alternatively, one has (see \cite{BF03,FH})
\begin{equation}\label{eq:PE2}
P(r;l)=E_2^{\mathrm{hard}}(4r;l),
\end{equation}
where
\begin{equation}\label{def:E2}
E_2^{\mathrm{hard}}(s;\nu):=\det(I-K_\nu^{\Bes})\big|_{L^2(0,s)}
\end{equation}
with $K_\nu^{\Bes}$ being the Bessel kernel defined in \eqref{def:Beskernel} is the scaled hard-edge gap probability of LUE over $(0,s)$. Thus, by showing the hard-to-soft transition
\begin{equation}
\lim_{\nu \to \infty}E_2^{\mathrm{hard}}\left(\left(\nu-t(\nu/2)^{1/3}\right)^2;\nu\right)=F(t),
\end{equation}
Borodin and Forrester reclaimed \eqref{eq:LF} in \cite{BF03}.

An interesting question now is to improve \eqref{eq:LnF} and \eqref{eq:LF} by establishing the first few finite-size correction terms or the asymptotic expansion. This is also known as edgeworth expansions in the literature, and we particularly refer to \cite{B16,Choup,EGP,Karoui,ForTri19,PS16} for the relevant results of Laguerre ensembles. In the context of the distribution for the length of longest increasing subsequences, the relationship \eqref{eq:PE2} plays an important role in a recent work of Bornemann \cite{B23} among various studies toward this aim \cite{BJ,BFoCM,FM23}. Instead of working on the Fredholm determinant directly, the idea in \cite{B23} is to establish an expansion between the Bessel kernel and the Airy kernel, which can be lifted to trace class operators. It is the aim of this paper to resolve some conjectures posed therein.

To proceed, we set
\begin{equation}\label{def:h}
h_{\nu}:=2^{-\frac 13} \nu^{-\frac 23},
\end{equation}
and define, as in \cite{B23}, the symmetrically transformed Bessel kernel
\begin{align}\label{def:inducedkernel}
\hat K_{\nu}^{\mathrm{Bes}}(x,y):=\sqrt{\phi_{\nu}'(x) \phi_{\nu}'(y)}K_{\nu}^{\mathrm{Bes}}(\phi_{\nu}(x),\phi_{\nu}(y)),
\end{align}
where
\begin{equation}\label{def:phi}
\phi_{\nu}(t):= \nu^2(1-h_{\nu}t)^2.
\end{equation}

Our main result is stated as follows.
\begin{theorem}\label{th:1}
With $\hat K_{\nu}^{\mathrm{Bes}}(x,y)$ defined in \eqref{def:inducedkernel}, we have, for any $\mathfrak{m}\in\mathbb{N}$,
\begin{align}\label{h-s}
\hat K_{\nu}^{\mathrm{Bes}}(x,y) = K^{\Ai}(x,y) + \sum_{j=1}^{\mathfrak{m}} K_j(x,y) h_{\nu}^j+h_{\nu}^{\mathfrak{m}+1} \cdot \Boh\left(e^{-(x+y)}\right), \qquad h_{\nu} \to 0^+,
\end{align}
uniformly valid for $t_0 \le x, y <h_{\nu}^{-1}$ with $t_0$ being any fixed real number. Preserving uniformity, the expansion can be repeatedly differentiated w.r.t. the variable $x$ and $y$. Here, $K^{\Ai}$ is the Airy kernel given in \eqref{def:KAi} and
\begin{align}\label{def:Kj}
K_j(x,y)=\sum_{\kappa, \lambda \in \{0,1\}} p_{j, \kappa \lambda}(x,y) \Ai^{(\kappa)}(x) \Ai^{(\lambda)}(y)
\end{align}
with $p_{j, \kappa \lambda} (x,y)$ being polynomials in $x$ and $y$. Moreover, we have
\begin{align}\label{def:K1}
K_1(x,y) &= \frac{1}{10}\left(-3(x^2+xy+y^2)\Ai(x)\Ai(y)+2(\Ai(x)\Ai'(y)+\Ai'(x)\Ai(y))\right. \nonumber \\
&\left.~~~ +3(x+y)\Ai'(x)\Ai'(y)\right),
\end{align}
and
\begin{align}\label{def:K2}
K_2(x,y) &= \frac{1}{1400}\left((56-235(x^2+y^2)-319xy(x+y))\Ai(x)\Ai(y)\right. \nonumber \\
&~~~ +(63(x^4+x^3y-x^2y^2-xy^3-y^4)-55x+239y)\Ai(x)\Ai'(y) \nonumber \\
&~~~ +(63(y^4+xy^3-x^2y^2-x^3y-x^4)-55y+239x)\Ai'(x)\Ai(y)\nonumber \\
&\left.~~ +(340(x^2+y^2)+256xy)\Ai'(x)\Ai'(y)\right).
\end{align}
\end{theorem}

Based on a uniform version of transient asymptotic expansion of Bessel functions \cite{Olver52}, the above theorem is stated in \cite{B23} under the condition that $0 \le \mathfrak{m} \le \mathfrak{m}_* =100$, where the upper bound $\mathfrak{m}_*$ is obtained through a numerical inspection. It is conjectured therein that \eqref{h-s} is valid without such a restriction, Theorem \ref{th:1} thus gives a confirm answer to this conjecture.


As long as the Bessel kernel admits an expansion of the form \eqref{h-s}, it is generally believed that one can lift the expansion to the associated Fredholm determinants. By carefully estimating trace norms in terms of kernel bounds, this is rigorously established in \cite[Theorem 2.1]{B23} for the perturbed Airy kernel determinants, which allows us to obtain the following  hard-to-soft edge transition expansion with the aid of Theorem \ref{th:1}.
\begin{corollary}
With $E_2^{\mathrm{hard}}$ defined in \eqref{def:E2}, we have, for any $\mathfrak{m}\in\mathbb{N}$,
\begin{align}\label{eq:hard-to-softexpan}
E_2^{\textrm{hard}}(\phi_{\nu}(t); \nu) = F(t) + \sum_{j=1}^{\mathfrak{m}} F_j(t) h_{\nu}^{j}+h_{\nu}^{\mathfrak{m}+1} \cdot \Boh\left(e^{-3t/2}\right),\qquad h_{\nu} \to 0^+,
\end{align}
uniformly valid for $t_0 \leq t < h_{\nu}^{-1}$ with $t_0$ being any fixed real number. Preserving uniformity, the expansion can be repeatedly differentiated w.r.t. the variable $t$. Here, $F$ denotes the Tracy-Widom distribution \eqref{def:F} and $F_j$ are certain smooth functions.
\end{corollary}
Again the above result is stated in \cite{B23} under a restriction on the number of summation but with explicit expressions of $F_1$ and $F_2$ in terms of the derivatives of $F$. The expansion \eqref{eq:hard-to-softexpan} serves as a preparatory step in establishing the expansion of the limit law \eqref{eq:LnF} in \cite[Theorem 5.1]{B23}. Finally, we also refer to \cite{CL23} for exponential moments, central limit theorems and rigidity of the hard-to-soft edge transition.

The rest of this paper is devoted to the proof of Theorem \ref{th:1}. The difficulty of using transient asymptotic expansion of the Bessel functions for large order to prove Theorem \ref{th:1} lies in checking the divisibility of a certain sequence of polynomials. Indeed, it was commented in \cite{B23} that one probably needs some hidden symmetry of the coefficients in the expansion. The approach we adopt here, however, is based on a Riemann-Hilbert (RH) characterization of the Bessel kernel, as described in Section \ref{sec:bessel}. By performing a Deift-Zhou nonlinear steepest descent analysis \cite{Deift1993} to the associated RH problem in Section \ref{sec:rhp}, the initial RH problem will be transformed into a small norm problem, for which one can find a uniform estimate. The proof of Theorem \ref{th:1} is an outcome of our asymptotic analysis, which is presented in Section \ref{sec:proof}. Our analysis also allows us to calculate the polynomial coefficients $p_{j, \kappa \lambda}$ of the expansion kernels $K_j(x,y)$ in a systematic way and to reproduce asymptotic expansion of the Bessel functions $J_{\nu}$ for large order and large argument; see Remark \ref{rk:coeffcal} and Appendix \ref{ap:B} below for details. In addition, we are confident that our methodology can be adapted to explore a variety of Edgeworth expansions arising from random matrix theory and beyond.


\begin{notation}
 Throughout this paper, the following notations are frequently used.
\begin{itemize}
  \item If $A$ is a matrix, then $(A)_{ij}$ stands for its $(i,j)$-th entry. We use $I$ to denote a $2\times2$ identity matrix.

  \item As usual, the three Pauli matrices $\{\sigma_j\}_{j=1}^3$ are defined by
\begin{equation}\label{def:Pauli}
\sigma_1=\begin{pmatrix}
           0 & 1 \\
           1 & 0
        \end{pmatrix},
        \qquad
        \sigma_2=\begin{pmatrix}
        0 & -\ii \\
        \ii & 0
        \end{pmatrix},
        \qquad
        \sigma_3=
        \begin{pmatrix}
        1 & 0 \\
         0 & -1
         \end{pmatrix}.
\end{equation}
\end{itemize}
\end{notation}

\section{An RH characterization of the Bessel kernel}\label{sec:bessel}
Our starting point is the following Bessel parametrix \cite{K2004}, which particularly characterizes the Bessel kernel $K_{\nu}^{\textrm{Bes}}(x,y)$ in \eqref{def:Beskernel}.

For $z\in \mathbb{C} \setminus [0, +\infty)$ and $\nu>-1$, we set
\begin{align}\label{def:psinu}
\Psi_{\nu}(z)=\sqrt{\pi} e^{-\frac{\pi \ii}{4}} \begin{pmatrix}
I_{\nu}\left((-z)^{\frac 12}\right) & -\frac{\ii}{\pi} K_{\nu}\left((-z)^{\frac 12}\right)\\
(-z)^{\frac 12}I_{\nu}'\left((-z)^{\frac 12}\right) & -\frac{\ii}{\pi}(-z)^{\frac 12} K'_{\nu}\left((-z)^{\frac 12}\right)
\end{pmatrix},
\end{align}
where $I_{\nu}$ and $K_{\nu}$ denote the modified Bessel functions of order $\nu$ (see \cite{DLMF}) and define
\begin{align}\label{def:Psi}
\Psi(z; \nu) = \Psi_{\nu}(z) \begin{cases}
\begin{pmatrix}
1 & 0\\
-e^{-\pi \ii \nu} & 1
\end{pmatrix}, & \qquad \arg z \in (0, \frac{\pi}{3}),\\
I, & \qquad \arg z \in (\frac{\pi}{3}, \frac{5 \pi}{3}),\\
\begin{pmatrix}
1 & 0\\
e^{\pi \ii \nu} & 1
\end{pmatrix}, & \qquad \arg z \in (\frac{5\pi}{3}, 2 \pi).
\end{cases}
\end{align}
Then $\Psi(z):=\Psi(z; \nu)$ satisfies the following RH problem.
\begin{rhp}\label{rhp:Psi}
\hfill
\begin{itemize}
\item[\rm (a)] $\Psi(z)$ is defined and  analytic for $z \in \mathbb{C} \setminus \left\{\cup_{j=1}^3 \Gamma_j \cup \{0\} \right\}$, where
\begin{align}\label{def:Gammai}
\Gamma_1:=e^{\frac{\pi \ii}{3}}(0, +\infty), \qquad \Gamma_2:=(0, +\infty), \qquad \Gamma_3:=e^{-\frac{\pi \ii}{3}}(0, +\infty);
\end{align}
see Figure \ref{fig:Psi} for an illustration.
\item[\rm (b)] For $z \in \Gamma_j$, $j=1, 2, 3$, the limiting values of $\Psi$ exist and satisfy the jump condition
\begin{align}
\Psi_+(z)=\Psi_-(z) \begin{cases}
\begin{pmatrix}
1 & 0\\
e^{-\pi \ii \nu} & 1
\end{pmatrix}, & \qquad z \in \Gamma_1,\\
\begin{pmatrix}
0 & 1\\
-1 & 0
\end{pmatrix}, & \qquad z \in \Gamma_2,\\
\begin{pmatrix}
1 & 0\\
e^{\pi \ii \nu} & 1
\end{pmatrix}, & \qquad z \in \Gamma_3.
\end{cases}
\end{align}
\item[\rm (c)] As $z \to \infty$, we have
\begin{align}
\Psi(z) = \begin{pmatrix}
1 & 0\\
-\frac{4 \nu^2+3}{8} & 1
\end{pmatrix} \left(I+\Boh(z^{-1})\right) (-z)^{-\frac{1}{4} \sigma_3} \frac{1}{\sqrt{2}} \begin{pmatrix}
1 & -1\\
1 & 1
\end{pmatrix} e^{((-z)^{1/2}-\frac{\pi \ii}{4}) \sigma_3},
\end{align}
where $\sigma_3$ is defined in \eqref{def:Pauli}, the branch cuts of $(-z)^{\pm \frac 14}$ and $(-z)^{\pm \frac 12}$ are chosen along $[0,+\infty)$.
\item[\rm (d)] As $z \to 0$, we have, for $\nu \notin \mathbb{Z}$,
\begin{align}
\Psi(z) = \widehat \Psi(z) (-z)^{\frac{\nu}{2} \sigma_3}\begin{pmatrix}
1 & \frac{\ii}{2 \sin \pi \nu}\\
0 & 1
\end{pmatrix}
\begin{cases}
\begin{pmatrix}
1 & 0\\
-e^{-\pi \ii \nu} & 1
\end{pmatrix}, & \qquad \arg z \in (0, \frac{\pi}{3}),\\
I, & \qquad \arg z \in (\frac{\pi}{3}, \frac{5 \pi}{3}),\\
\begin{pmatrix}
1 & 0\\
e^{\pi \ii \nu} & 1
\end{pmatrix}, & \qquad \arg z \in (\frac{5\pi}{3}, 2 \pi),
\end{cases}
\end{align}
and for $\nu \in \mathbb{Z}$,
\begin{align}
\Psi(z) = \widehat \Psi(z) (-z)^{\frac{\nu}{2} \sigma_3}\begin{pmatrix}
1 & -\frac{e^{\pi \ii \nu}}{2 \pi \ii} \ln (-z)\\
0 & 1
\end{pmatrix}
\begin{cases}
\begin{pmatrix}
1 & 0\\
-e^{-\pi \ii \nu} & 1
\end{pmatrix}, & \qquad \arg z \in (0, \frac{\pi}{3}),\\
I, & \qquad \arg z \in (\frac{\pi}{3}, \frac{5 \pi}{3}),\\
\begin{pmatrix}
1 & 0\\
e^{\pi \ii \nu} & 1
\end{pmatrix}, & \qquad \arg z \in (\frac{5\pi}{3}, 2 \pi),
\end{cases}
\end{align}
where $\what \Psi(z)$ is analytic at $z=0$ and we choose principal branches for $(-z)^{\pm \frac{\nu}{2}}$ and $\ln (-z)$.
\end{itemize}
\end{rhp}

%
%
%
%
%
%
%
%
%
%
%
%
%

\begin{figure}[ht]
\begin{center}

\tikzset{every picture/.style={line width=0.75pt}} 

\begin{tikzpicture}[x=0.75pt,y=0.75pt,yscale=-1,xscale=1]

\draw    (240,370.33) -- (409,370.5) ;
\draw [shift={(329.5,370.42)}, rotate = 180.06] [fill={rgb, 255:red, 0; green, 0; blue, 0 }  ][line width=0.08]  [draw opacity=0] (8.93,-4.29) -- (0,0) -- (8.93,4.29) -- cycle    ;
\draw    (240,370.33) -- (380.67,450.25) ;
\draw [shift={(314.68,412.76)}, rotate = 209.6] [fill={rgb, 255:red, 0; green, 0; blue, 0 }  ][line width=0.08]  [draw opacity=0] (8.93,-4.29) -- (0,0) -- (8.93,4.29) -- cycle    ;
\draw    (240,370.33) -- (379,290.5) ;
\draw [shift={(313.84,327.93)}, rotate = 150.13] [fill={rgb, 255:red, 0; green, 0; blue, 0 }  ][line width=0.08]  [draw opacity=0] (8.93,-4.29) -- (0,0) -- (8.93,4.29) -- cycle    ;
\draw  [fill={rgb, 255:red, 0; green, 0; blue, 0 }  ,fill opacity=1 ] (242.08,370.3) .. controls (242.1,371.45) and (241.18,372.4) .. (240.03,372.42) .. controls (238.88,372.43) and (237.93,371.52) .. (237.92,370.36) .. controls (237.9,369.21) and (238.82,368.27) .. (239.97,368.25) .. controls (241.12,368.23) and (242.07,369.15) .. (242.08,370.3) -- cycle ;

\draw (230.33,373.67) node [anchor=north west][inner sep=0.75pt]   [align=left] {0};
\draw (383.33,279.67) node [anchor=north west][inner sep=0.75pt]   [align=left] {$\Gamma_1$};
\draw (413.83,362.67) node [anchor=north west][inner sep=0.75pt]   [align=left] {$\Gamma_2$};
\draw (383.33,445.17) node [anchor=north west][inner sep=0.75pt]   [align=left] {$\Gamma_3$};

\end{tikzpicture}

\caption{The jump contours of the RH problem for $\Psi$.}
\label{fig:Psi}

\end{center}
\end{figure}
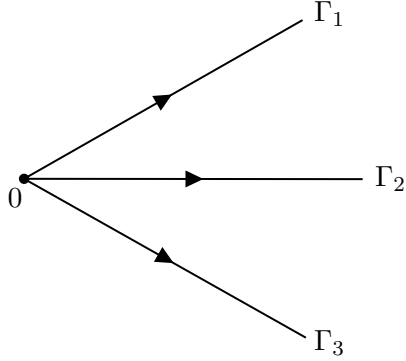

The Bessel kernel then admits the following representation:
\begin{align}\label{def:bessel}
K_{\nu}^{\Bes}(x, y) = \frac{1}{2 \pi \ii (x-y)} \begin{pmatrix}
-e^{-\frac{\pi \ii \nu}{2}} & e^{\frac{\pi \ii \nu}{2}}
\end{pmatrix}
\Psi_{+}(y)^{-1}\Psi_{+}(x)\begin{pmatrix}
e^{\frac{\pi \ii \nu}{2}} \\ e^{-\frac{\pi \ii \nu}{2}}
\end{pmatrix},
\end{align}
where the $+$ sign means taking limits from the upper half-plane.

\section{Asymptotic analysis of the RH problem for $\Psi$ with large $\nu$}
\label{sec:rhp}

In this section, we will perform a Deift-Zhou steepest descent analysis \cite{Deift1999} for the RH problem
for $\Psi$. It consists of a series of explicit and invertible transformations and the final goal is to
arrive at an RH problem tending to the identity matrix as $\nu \to +\infty$. Without loss of generality, we may assume that $\nu>0$ in what follows.

\subsection{First transformation: $\Psi \to Y$}
Due to the scaled variable \eqref{def:phi}, the first transformation is a scaling.
In addition, we multiply some constant matrices from the left to simplify the asymptotic behavior at infinity.

Define the matrix-valued function
\begin{align}\label{def:Y}
Y(z) = \nu^{\frac 12 \sigma_3}  \begin{pmatrix}
1 & 0\\
\frac{4 \nu^2+3}{8} & 1
\end{pmatrix}\Psi(\nu^2 z).
\end{align}
It is then readily seen from RH problem \ref{rhp:Psi} that $Y$ satisfies the following RH problem.
\begin{rhp}
\hfill
\begin{itemize}
\item[\rm (a)] $Y(z)$ is defined and  analytic for $z \in \mathbb{C} \setminus \left\{\cup_{j=1}^3 \Gamma_j \cup \{0\} \right\}$, where $\Gamma_i$, $i=1,2,3$, is defined in \eqref{def:Gammai}.
\item[\rm (b)] $Y(z)$ satisfies the jump condition
\begin{align}
Y_+(z)=Y_-(z) \begin{cases}
\begin{pmatrix}
1 & 0\\
e^{-\pi \ii \nu} & 1
\end{pmatrix}, & \qquad z \in \Gamma_1,\\
\begin{pmatrix}
0 & 1\\
-1 & 0
\end{pmatrix}, & \qquad z \in \Gamma_2,\\
\begin{pmatrix}
1 & 0\\
e^{\pi \ii \nu} & 1
\end{pmatrix}, & \qquad z \in \Gamma_3.
\end{cases}
\end{align}
\item[\rm (c)] As $z \to \infty$, we have
\begin{align}
Y(z) =  \left(I+\Boh(z^{-1})\right) (-z)^{-\frac{1}{4} \sigma_3} \frac{1}{\sqrt{2}} \begin{pmatrix}
1 & -1\\
1 & 1
\end{pmatrix} e^{(\nu(-z)^{1/2}-\frac{\pi \ii}{4}) \sigma_3},
\end{align}
where we take the principal branch for fractional exponents.
\end{itemize}
\end{rhp}

\subsection{Second transformation: $Y \to T$}
In the second transformation we apply contour deformations. The rays $\Gamma_1$ and $\Gamma_3$ emanating from the origin are replaced by their
parallel lines $\widetilde \Gamma_1$ and $\widetilde \Gamma_3$ emanating from $1$. Let $\textrm{I}$ and $\textrm{II}$ be two regions bounded
by $\Gamma_1\cup \widetilde \Gamma_1 \cup [0,1]$ and $\Gamma_3 \cup \widetilde \Gamma_3 \cup [0,1]$, respectively; see Figure \ref{fig:T} for an illustration. We now define
\begin{align}\label{def:T}
T(z) = Y(z) \begin{cases}
\begin{pmatrix}
1 & 0\\
e^{-\pi \ii \nu} & 1
\end{pmatrix}, & \qquad z \in \textrm{I},\\
\begin{pmatrix}
1 & 0\\
-e^{\pi \ii \nu} & 1
\end{pmatrix}, & \qquad z \in \textrm{II}, \\
I, & \qquad \textrm{elsewhere.}
\end{cases}
\end{align}

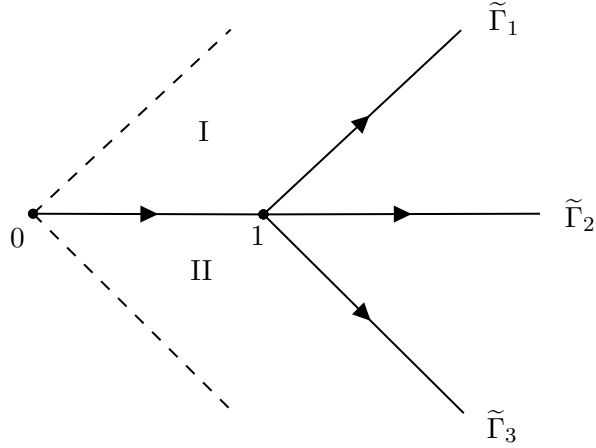
\begin{figure}[ht]
\begin{center}

\tikzset{every picture/.style={line width=0.75pt}} 

\begin{tikzpicture}[x=0.75pt,y=0.75pt,yscale=-1,xscale=1]

\draw    (100,118) -- (215,118.34) -- (353,118) ;
\draw [shift={(162.5,118.19)}, rotate = 180.17] [fill={rgb, 255:red, 0; green, 0; blue, 0 }  ][line width=0.08]  [draw opacity=0] (8.93,-4.29) -- (0,0) -- (8.93,4.29) -- cycle    ;
\draw [shift={(289,118.16)}, rotate = 179.86] [fill={rgb, 255:red, 0; green, 0; blue, 0 }  ][line width=0.08]  [draw opacity=0] (8.93,-4.29) -- (0,0) -- (8.93,4.29) -- cycle    ;
\draw    (215,118.34) -- (315,218.34) ;
\draw [shift={(268.53,171.88)}, rotate = 225] [fill={rgb, 255:red, 0; green, 0; blue, 0 }  ][line width=0.08]  [draw opacity=0] (8.93,-4.29) -- (0,0) -- (8.93,4.29) -- cycle    ;
\draw  [dash pattern={on 4.5pt off 4.5pt}]  (100,118) -- (200,218) ;
\draw  [dash pattern={on 4.5pt off 4.5pt}]  (100,118) -- (198.67,25.33) ;
\draw    (215,118.34) -- (313.66,25.67) ;
\draw [shift={(267.97,68.58)}, rotate = 136.8] [fill={rgb, 255:red, 0; green, 0; blue, 0 }  ][line width=0.08]  [draw opacity=0] (8.93,-4.29) -- (0,0) -- (8.93,4.29) -- cycle    ;

\draw (87,124) node [anchor=north west][inner sep=0.75pt]   [align=left] {$0$};
\draw (207,123) node [anchor=north west][inner sep=0.75pt]   [align=left] {$1$};
\draw (326,10) node [anchor=north west][inner sep=0.75pt]   [align=left] {$\widetilde\Gamma_1$};
\draw (364,109) node [anchor=north west][inner sep=0.75pt]   [align=left] {$\widetilde\Gamma_2$};
\draw (325,215) node [anchor=north west][inner sep=0.75pt]   [align=left] {$\widetilde\Gamma_3$};
\draw (181,70) node [anchor=north west][inner sep=0.75pt]   [align=left] {I};
\draw (177,140) node [anchor=north west][inner sep=0.75pt]   [align=left] {II};

\draw [fill={rgb, 255:red, 0; green, 0; blue, 0 }  ,fill opacity=1 ]  (215, 118.34) circle [x radius= 2, y radius= 2]   ;
\draw [fill={rgb, 255:red, 0; green, 0; blue, 0 }  ,fill opacity=1 ]  (215, 118.34) circle [x radius= 2, y radius= 2]   ;
\draw [fill={rgb, 255:red, 0; green, 0; blue, 0 }  ,fill opacity=1 ]  (100, 118) circle [x radius= 2, y radius= 2]   ;
\draw [fill={rgb, 255:red, 0; green, 0; blue, 0 }  ,fill opacity=1 ]  (100, 118) circle [x radius= 2, y radius= 2]   ;
\draw [fill={rgb, 255:red, 0; green, 0; blue, 0 }  ,fill opacity=1 ]  (215, 118.34) circle [x radius= 2, y radius= 2]   ;
\draw [fill={rgb, 255:red, 0; green, 0; blue, 0 }  ,fill opacity=1 ]  (215, 118.34) circle [x radius= 2, y radius= 2]   ;
\end{tikzpicture}

   \caption{The jump contour $\Gamma_T$ of the RH problem for $T$.}
   \label{fig:T}
\end{center}
\end{figure}
It is easily seen the following RH problem for $T$.
\begin{rhp}\label{rhp:T}
\hfill
\begin{itemize}
\item [\rm (a)] $T(z)$ is defined and analytic in $\mathbb{C} \setminus \Gamma_T$, where
\begin{equation}\label{def:GammaT}
\Gamma_T:= \cup_{i=1}^3\widetilde \Gamma_i \cup [0,1],
\end{equation}
with
\begin{align}
\widetilde\Gamma_1:=e^{\frac{\pi \ii}{3}}(1, +\infty), \qquad \widetilde\Gamma_2:=(1, +\infty), \qquad \widetilde\Gamma_3:=e^{-\frac{\pi \ii}{3}}(1, +\infty);
\end{align}
see Figure \ref{fig:T}.
\item [\rm (b)] $T(z)$ satisfies the jump condition
\begin{align}
T_+(z) = T_-(z) \begin{cases}
\begin{pmatrix}
1 & 0\\
e^{-\pi \ii \nu} & 1
\end{pmatrix}, & \qquad z \in \widetilde \Gamma_1,\\
\begin{pmatrix}
0 & 1\\
-1 & 0
\end{pmatrix}, & \qquad z \in \widetilde \Gamma_2,\\
\begin{pmatrix}
1 & 0\\
e^{\pi \ii \nu} & 1
\end{pmatrix}, & \qquad z \in \widetilde \Gamma_3,\\
\begin{pmatrix}
e^{-\pi \ii \nu} & 1\\
0 & e^{\pi \ii \nu}
\end{pmatrix}, & \qquad z \in (0,1).
\end{cases}
\end{align}
\item[\rm (c)] As $z \to \infty$, we have
\begin{align}
T(z) =  \left(I+\Boh(z^{-1})\right) (-z)^{-\frac{1}{4} \sigma_3} \frac{1}{\sqrt{2}} \begin{pmatrix}
1 & -1\\
1 & 1
\end{pmatrix} e^{(\nu(-z)^{1/2}-\frac{\pi \ii}{4}) \sigma_3}.
\end{align}
\end{itemize}
\end{rhp}

\subsection{Third transformation: $T \to S$}
In order to normalize the behavior at infinity, we apply the third transformation $T \to S$ by introducing the so-called $g$-function:
\begin{align}\label{def:g}
g(z) := - (1-z)^{\frac 12} + \frac 12 \ln \left(\frac{1 +  (1-z)^{1/2}}{1- (1-z)^{1/2}}\right)\pm \frac{\pi \ii}{2}, \quad \pm \Im z > 0.
\end{align}
As before, we take a cut along $[1, +\infty)$ for $(1-z)^{\frac 12}$. The following proposition of $g$ is immediate from its definition.
\begin{proposition}\label{pro:g}
\hfill
\begin{itemize}
\item [\rm (i)] The function $g(z)$ is analytic in $\mathbb{C} \setminus [0, +\infty)$.
\item [\rm (ii)] For $z \in (-\infty, 0)$, we have
\begin{align}
g(z) = -\sqrt{1-z} + \frac{1}{2} \ln \left(\frac{1+\sqrt{1-z}}{\sqrt{1-z}-1}\right).
\end{align}
\item [\rm (iii)]For $z \in (0,1)$, we have
\begin{align}\label{g01}
g_{\pm}(z) = -\sqrt{1-z} + \frac{1}{2} \ln \left(\frac{1+\sqrt{1-z}}{1-\sqrt{1-z}}\right)\pm \frac{\pi \ii}{2}.
\end{align}
\item[\rm (iv)] For $z \in (1, +\infty)$, we have
\begin{align}
g_{\pm}(z) = \pm \ii \sqrt{z-1} \pm \frac{\ii }{2} \arg \left(\frac{1- \ii \sqrt{z-1}}{1+\ii\sqrt{z-1}}\right)\pm \frac{\pi \ii}{2}.
\end{align}
\item [\rm (v)]As $z \to \infty$, we have
\begin{align}
g(z) = - \left(-z\right)^{\frac 12} + \frac 12 \left(-z\right)^{-\frac 12} + \Boh(z^{-1}).
\end{align}
\end{itemize}
\end{proposition}

\begin{figure}[h]
 \centering
  \includegraphics[scale=.6]{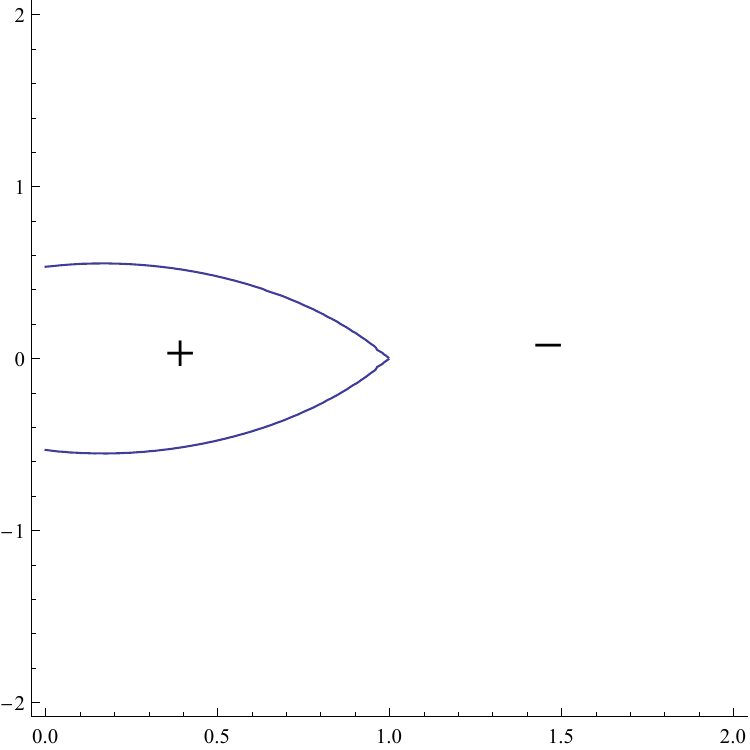}
  \caption{Image of $\Re g$: the solid line is the contour of $\Re{g(z)}=0$, the ``$-$" sign is the region where $\Re{g(z)}<0$ and the ``$+$" sign shows the region where $\Re{g(z)}>0$.}
  \label{fig:g}
\end{figure}

By setting
\begin{align}\label{def:S}
S(z) = T(z) e^{\nu g(z) \sigma_3},
\end{align}
it is readily seen from Proposition \ref{pro:g} and RH problem \ref{rhp:T} that $S$ satisfies the RH problem as follows.
\begin{rhp}
\hfill
\begin{itemize}
\item [\rm (a)] $S(z)$ is defined and analytic in $\mathbb{C} \setminus \Gamma_T$, where the contour $\Gamma_T$ is defined in \eqref{def:GammaT}.
\item [\rm (b)] $S(z)$ satisfies the jump condition
\begin{align}
S_+(z) = S_-(z) \begin{cases}
\begin{pmatrix}
1 & 0\\
e^{2 \nu g(z)-\pi \ii \nu} & 1
\end{pmatrix}, & \qquad z \in \widetilde \Gamma_1,\\
\begin{pmatrix}
0 & 1\\
-1 & 0
\end{pmatrix}, & \qquad z \in \widetilde \Gamma_2,\\
\begin{pmatrix}
1 & 0\\
e^{2 \nu g(z)+\pi \ii \nu} & 1
\end{pmatrix}, & \qquad z \in \widetilde \Gamma_3,\\
\begin{pmatrix}
1 & e^{2 \nu \left(\sqrt{1-z} - \frac 12 \ln \left(\frac{1+\sqrt{1-z}}{1-\sqrt{1-z}}\right)\right)}\\
0 & 1
\end{pmatrix}, & \qquad z \in (0,1).
\end{cases}
\end{align}
\item[\rm (c)] As $z \to \infty$, we have
\begin{align}
S(z) =  \left(I+\Boh(z^{-1})\right) (-z)^{-\frac{1}{4} \sigma_3} \frac{1}{\sqrt{2}} \begin{pmatrix}
1 & -1\\
1 & 1
\end{pmatrix} e^{-\frac{\pi \ii}{4} \sigma_3}.
\end{align}
\end{itemize}
\end{rhp}

\subsection{Global parametrix}
As $\nu \to +\infty$, from the image of $\Re g$ depicted in Figure \ref{fig:g}, we conclude that all the jump matrices of $S$ tend to $I$ exponentially fast except for that along $(1, +\infty)$. Ignoring the exponential small terms in the jump matrices for $S$, we come to the following global parametrix.
\begin{rhp}
\hfill
\begin{itemize}
\item [\rm (a)] $N(z)$ is defined and analytic in $\mathbb{C} \setminus [1, +\infty)$.
\item [\rm (b)] $N(z)$ satisfies the jump condition
\begin{align}\label{jump:N}
N_+(z) = N_-(z) \begin{pmatrix}
0 & 1\\
-1 & 0
\end{pmatrix}, \qquad z\in[1,+\infty).
\end{align}
\item[\rm (c)] As $z \to \infty$, we have
\begin{align}
N(z) =  \left(I+\Boh(z^{-1})\right) (-z)^{-\frac{1}{4} \sigma_3} \frac{1}{\sqrt{2}} \begin{pmatrix}
1 & -1\\
1 & 1
\end{pmatrix} e^{-\frac{\pi \ii}{4} \sigma_3}.
\end{align}
\end{itemize}
\end{rhp}
An explicit solution to the RH problem for $N$ is given by
\begin{align}\label{def:N}
N(z) =(1-z)^{-\frac{1}{4} \sigma_3} \frac{1}{\sqrt{2}} \begin{pmatrix}
1 & -1\\
1 & 1
\end{pmatrix} e^{-\frac{\pi \ii}{4} \sigma_3}.
\end{align}

\subsection{Local parametrix}
Since the jump matrices for $S$ and $N$ are not uniformly close to each other near the point $z=1$, we next construct local parametrix near this point. In a disc $D(1, \varepsilon)$ centered at 1 with certain fixed radius $0<\varepsilon<1$, we seek a $2 \times 2$ matrix-valued function $P(z)$ satisfying an RH problem as follows.
\begin{rhp}\label{rhp:P1}
\hfill
\begin{itemize}
\item [\rm (a)] $P(z)$ is defined and analytic in $D(1, \varepsilon) \setminus \Gamma_T$, where the contour $\Gamma_T$ is defined in \eqref{def:GammaT}.
\item [\rm (b)] $P(z)$ satisfies the jump condition
\begin{align}\label{jump:P1}
P_+(z) = P_-(z) \begin{cases}
\begin{pmatrix}
1 & 0\\
e^{2 \nu g(z)-\pi \ii \nu} & 1
\end{pmatrix}, & \qquad z \in D(1, \varepsilon) \cap \widetilde \Gamma_1,\\
\begin{pmatrix}
0 & 1\\
-1 & 0
\end{pmatrix}, & \qquad z \in D(1, \varepsilon) \cap \widetilde \Gamma_2,\\
\begin{pmatrix}
1 & 0\\
e^{2 \nu g(z)+\pi \ii \nu} & 1
\end{pmatrix}, & \qquad z \in D(1, \varepsilon) \cap \widetilde \Gamma_3,\\
\begin{pmatrix}
1 & e^{2 \nu \left(\sqrt{1-z} - \frac 12 \ln \left(\frac{1+\sqrt{1-z}}{1-\sqrt{1-z}}\right)\right)}\\
0 & 1
\end{pmatrix}, & \qquad z \in D(1, \varepsilon) \cap (0,1).
\end{cases}
\end{align}
\item[\rm (c)] As $\nu \to \infty$, we have
\begin{align}\label{matching}
P(z) =  \left(I+\Boh(\nu^{-1})\right) N(z), \qquad z\in \partial D(1,\varepsilon),
\end{align}
where $N$ is given in \eqref{def:N}.
\end{itemize}
\end{rhp}
This local parametrix can be constructed by using the Airy parametrix $\Phi^{(\Ai)}$ introduced in \hyperref[airy]{Appendix A}. To do this, we introduce the function:
\begin{align}\label{def:f}
 f(z) &= \left(\frac{3g(z)}{2} \mp \frac{3\pi \ii}{4}\right)^{\frac 23}, \qquad \pm \Im z > 0 \\
 &= -2^{-\frac 23} (z-1) \left(1 - \frac{2}{5} (z-1) +\frac{43}{175}(z-1)^2 +\Boh \left((z-1)^3\right)\right), \qquad z \to 1.\nonumber
\end{align}
It is easily obtained that
\begin{align}
f(1)=0, \qquad f'(1)=-2^{-\frac 23}<0.
\end{align}
We then set
\begin{align}\label{def:P1}
P(z) = E(z) \Phi^{(\Ai)}\left(\nu^{\frac 23}f(z)\right) e^{\nu \left(g(z) \mp \frac{\pi \ii}{2}\right)\sigma_3} \sigma_3
\end{align}
with
\begin{align}\label{def:E}
E(z) := N(z) \sigma_3 \frac{1}{\sqrt{2}} \begin{pmatrix}
1 & -\ii\\
-\ii & 1
\end{pmatrix}f(z)^{\frac 14 \sigma_3}\nu^{\frac 16 \sigma_3}.
\end{align}

\begin{proposition}\label{pro:P}
The matrix-valued function $P(z)$ defined in \eqref{def:P1} solves RH problem \ref{rhp:P1}.
\end{proposition}
\begin{proof}
We first show the prefactor $E(z)$ is analytic near $z=1$. According to its definition in \eqref{def:E}, the only possible jump is on $(1, 1+ \varepsilon)$. It follows from \eqref{jump:N} and \eqref{def:f} that, if $z \in (1, 1+ \varepsilon)$,
\begin{align}
E_-(z)^{-1}E_+(z) &= \nu^{-\frac 16 \sigma_3}f_-(z)^{-\frac 14 \sigma_3}\frac{1}{\sqrt{2}} \begin{pmatrix}
1 & \ii\\
\ii & 1
\end{pmatrix} \sigma_3 \begin{pmatrix}
0 & 1\\
-1 & 0
\end{pmatrix}\sigma_3\frac{1}{\sqrt{2}} \begin{pmatrix}
1 & -\ii\\
-\ii & 1
\end{pmatrix}f_+(z)^{\frac 14 \sigma_3}\nu^{\frac 16 \sigma_3}\\
&=\nu^{-\frac 16 \sigma_3}f_-(z)^{-\frac 14 \sigma_3}\begin{pmatrix}
\ii & 0\\
0 & -\ii
\end{pmatrix}f_+(z)^{\frac 14 \sigma_3}\nu^{\frac 16 \sigma_3}=I,\nonumber
\end{align}
since $f_+(z)=e^{-2 \pi \ii}f_-(z)$ for $z \in (1, 1+ \varepsilon)$. Thus, $E(z)$ is analytic in $D(1,\varepsilon) \setminus \{1\}$. Note that, as $z \to 1$, by using \eqref{def:f},
\begin{align}\label{E1}
E(z)=(2h_{\nu})^{-\frac 14 \sigma_3}e^{-\frac{\pi \ii}{4} \sigma_3} \sigma_3 \left(I-\frac{\sigma_3}{10}(z-1)+\begin{pmatrix}
\frac{13}{280} & 0\\
0 & -\frac{51}{1400}
\end{pmatrix}(z-1)^2+\Boh((z-1)^3)\right)
\end{align}
with $h_{\nu}$ given in \eqref{def:h}, we conclude that $z=1$ is a removable singularity. The jump conditions of $P(z)$ in \eqref{jump:P1} follows directly from the analyticity of $E(z)$ and \eqref{jump:Airy}. Finally,
as $\nu \to +\infty$, we apply \eqref{infty:Ai} and obtain after a straightforward computation that
\begin{align}
P(z)N(z)^{-1} = N(z) \sigma_3 \left(I + \Boh (\nu^{-1})
\right) \sigma_3 N(z)^{-1}=I + \Boh (\nu^{-1}), \qquad z\in \partial D(1,\varepsilon).
\end{align}

This completes the proof of Proposition \ref{pro:P}.
\end{proof}

\subsection{Final transformation}
We define the final transformation
\begin{align}\label{def:R}
R(z)=\begin{cases}
S(z)P(z)^{-1}, & \qquad z \in D(1, \varepsilon),\\
S(z) N(z)^{-1}, & \qquad \textrm{elsewhere}.
\end{cases}
\end{align}
It is then readily seen that $R(z)$ satisfies the following RH problem.
\begin{rhp}\label{rhp:R}
\hfill
\begin{itemize}
\item[\rm (a)] $R(z)$ is defined and analytic in $\mathbb{C} \setminus \Gamma_R$, where
$$
\Gamma_R:=\partial D(1,\varepsilon) \cup \Gamma_T \setminus D(1,\varepsilon);
$$
see Figure \ref{fig:R} for an illustration.
\item[\rm (b)] $R(z)$ satisfies the jump condition
\begin{align}
R_+(z)=R_-(z) J_R(z), \qquad z\in \Gamma_R,
\end{align}
where
\begin{align}\label{jump:R}
J_R(z) = \begin{cases}
P(z)N(z)^{-1}, & \qquad z \in \partial D(1, \varepsilon),\\
N(z)S(z)N(z)^{-1}, & \qquad z \in \Gamma_R \setminus \partial D(1, \varepsilon).
\end{cases}
\end{align}
\item[\rm (c)] As $z \to \infty$, we have
\begin{align}
R(z) = I + \Boh(z^{-1}).
\end{align}
\end{itemize}
\end{rhp}

\begin{figure}[ht]
\begin{center}

\tikzset{every picture/.style={line width=0.75pt}} 

\begin{tikzpicture}[x=0.75pt,y=0.75pt,yscale=-1,xscale=1]

\draw   (256,121) .. controls (256,107.19) and (267.19,96) .. (281,96) .. controls (294.81,96) and (306,107.19) .. (306,121) .. controls (306,134.81) and (294.81,146) .. (281,146) .. controls (267.19,146) and (256,134.81) .. (256,121) -- cycle ;
\draw    (166,120) -- (256,121) ;
\draw [shift={(216,120.56)}, rotate = 180.64] [fill={rgb, 255:red, 0; green, 0; blue, 0 }  ][line width=0.08]  [draw opacity=0] (8.93,-4.29) -- (0,0) -- (8.93,4.29) -- cycle    ;
\draw    (297,140) -- (381,221) ;
\draw [shift={(342.6,183.97)}, rotate = 223.96] [fill={rgb, 255:red, 0; green, 0; blue, 0 }  ][line width=0.08]  [draw opacity=0] (8.93,-4.29) -- (0,0) -- (8.93,4.29) -- cycle    ;
\draw    (381,18) -- (297,101) ;
\draw [shift={(343.62,54.93)}, rotate = 135.34] [fill={rgb, 255:red, 0; green, 0; blue, 0 }  ][line width=0.08]  [draw opacity=0] (8.93,-4.29) -- (0,0) -- (8.93,4.29) -- cycle    ;
\draw  [fill={rgb, 255:red, 0; green, 0; blue, 0 }  ,fill opacity=1 ] (281,121) .. controls (281,120.26) and (281.6,119.67) .. (282.33,119.67) .. controls (283.07,119.67) and (283.67,120.26) .. (283.67,121) .. controls (283.67,121.74) and (283.07,122.33) .. (282.33,122.33) .. controls (281.6,122.33) and (281,121.74) .. (281,121) -- cycle ;
\draw  [fill={rgb, 255:red, 0; green, 0; blue, 0 }  ,fill opacity=1 ] (164,120) .. controls (164,119.26) and (164.6,118.67) .. (165.33,118.67) .. controls (166.07,118.67) and (166.67,119.26) .. (166.67,120) .. controls (166.67,120.74) and (166.07,121.33) .. (165.33,121.33) .. controls (164.6,121.33) and (164,120.74) .. (164,120) -- cycle ;
\draw    (281,96) -- (284,96) ;
\draw [shift={(287,96)}, rotate = 180] [fill={rgb, 255:red, 0; green, 0; blue, 0 }  ][line width=0.08]  [draw opacity=0] (8.93,-4.29) -- (0,0) -- (8.93,4.29) -- cycle    ;

\draw (280,124) node [anchor=north west][inner sep=0.75pt]   [align=left] {$1$};
\draw (386,14) node [anchor=north west][inner sep=0.75pt]   [align=left] {$\widetilde\Gamma_1$};
\draw (390,215) node [anchor=north west][inner sep=0.75pt]   [align=left] {$\widetilde\Gamma_3$};
\draw (159,123) node [anchor=north west][inner sep=0.75pt]   [align=left] {$0$};

\end{tikzpicture}

   \caption{The jump contours of the RH problem for $R$.}
   \label{fig:R}
\end{center}
\end{figure}
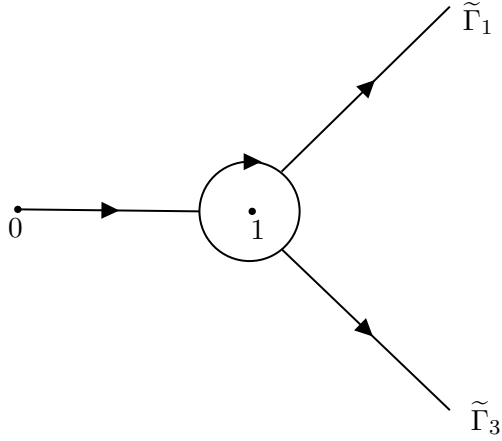

For $z \in \Gamma_R \setminus \partial D(1, \varepsilon)$, we have the estimate
\begin{equation}
J_R (z)=I + \Boh(e^{-c \nu}), \qquad \nu \to +\infty,
\end{equation}
for some constant $c>0$.

For $z \in \partial D(1, \varepsilon)$, substituting the full expansion \eqref{asy:Ai} of $\Phi^{({\Ai})}$ into \eqref{def:P1} and \eqref{jump:R}, we have
\begin{align}
J_R (z)\sim I + \sum_{k=1}^{\infty} J_{R,k}(z)h_{\nu}^{\frac{3k}{2}}, \qquad \nu \to +\infty,
\end{align}
where
\begin{align}\label{Jk}
J_{R,k}(z)=\begin{cases}
-\frac{3^{k}}{2^{k/2}f(z)^{3k/2}}\begin{pmatrix}
0 & (1-z)^{-\frac 12} \mathfrak{u}_k\\
(1-z)^{\frac 12} \mathfrak{v}_k & 0
\end{pmatrix}, & \quad \textrm{for odd $k$, }\\
\frac{3^{k}}{2^{k/2}f(z)^{3k/2}}\begin{pmatrix}
\mathfrak{u}_k & 0\\
0 & \mathfrak{v}_k
\end{pmatrix}, & \quad \textrm{for even $k$, }
\end{cases}
\end{align}
with $f(z)$ and $\mathfrak{u}_k$, $\mathfrak{v}_k$ given in \eqref{def:f} and \eqref{def:ukvk}, respectively.
By a standard argument \cite{Deift1999, Deift1993}, we conclude that, as $\nu \to +\infty$,
\begin{align}\label{eq:Rexp}
R(z) \sim I + \sum_{k=1}^{\infty} R_k(z)h_{\nu}^{\frac{3k}{2}}
\end{align}
uniformly for $z \in \mathbb{C} \setminus \Gamma_R$. A combination of \eqref{eq:Rexp} and RH problem \ref{rhp:R} shows that $R_1$ satisfies
\begin{rhp}
\hfill
\begin{itemize}
\item[\rm (a)] $R_1(z)$ is defined and analytic in $\mathbb{C} \setminus \partial D(1, \varepsilon)$.
\item[\rm (b)] $R_1(z)$ satisfies the jump condition
\begin{align}
R_{1,+}(z)=R_{1,-}(z) +J_{R,1}(z), \qquad z \in \partial D(1, \varepsilon),
\end{align}
where
\begin{align}\label{jump:R1}
J_{R,1}(z) =-\frac{\sqrt{2}}{48 f(z)^{3/2}} \begin{pmatrix}
0& 5(1-z)^{- \frac 12}\\
-7 (1-z)^{\frac 12} & 0
\end{pmatrix}.
\end{align}
\item[\rm (c)] As $z \to \infty$, we have
$R_1(z) = \Boh(z^{-1})$.
\end{itemize}
\end{rhp}
From the local behavior of $f(z)$ near $z=1$ given in \eqref{def:f}, we obtain that
\begin{align}
J_{R,1}(z)=\frac{1}{(z-1)^2}\begin{pmatrix}
0 & -\frac{5 \sqrt{2}}{24}\\
0& 0
\end{pmatrix}-\frac{1}{z-1} \begin{pmatrix}
0 & \frac{\sqrt{2}}{8}\\
\frac{7\sqrt{2}}{24}& 0
\end{pmatrix}-\begin{pmatrix}
0 & -\frac{\sqrt{2}}{70}\\
\frac{7\sqrt{2}}{40}& 0
\end{pmatrix} + \Boh(z-1).
\end{align}
By Cauchy's residue theorem, we have
\begin{align}\label{def:R1}
R_1(z) &= \frac{1}{2 \pi \ii} \oint_{\partial D(1, \varepsilon)}\frac{J_{R,1}(s)}{z-s} \ud s\\
&=\begin{cases}
\frac{1}{(z-1)^2}\begin{pmatrix}
0 & -\frac{5 \sqrt{2}}{24}\\
0& 0
\end{pmatrix}-\frac{1}{z-1} \begin{pmatrix}
0 & \frac{\sqrt{2}}{8}\\
\frac{7\sqrt{2}}{24}& 0
\end{pmatrix}, \quad & z \in \mathbb{C} \setminus D(1, \varepsilon),\\
\frac{1}{(z-1)^2}\begin{pmatrix}
0 &-\frac{5 \sqrt{2}}{24}\\
0& 0
\end{pmatrix}-\frac{1}{z-1} \begin{pmatrix}
0 & \frac{\sqrt{2}}{8}\\
\frac{7\sqrt{2}}{24}& 0
\end{pmatrix}-J_{R,1}(z), \quad & z \in D(1, \varepsilon).
\end{cases}\nonumber
\end{align}
Similarly, $R_2$ satisfies the following RH problem.
\begin{rhp}
\hfill
\begin{itemize}
\item[\rm (a)] $R_2(z)$ is defined and analytic in $\mathbb{C} \setminus \partial D(1, \varepsilon)$.
\item[\rm (b)]  $R_2(z)$ satisfies the jump condition
\begin{align}
R_{2,+}(z)=R_{2,-}(z) +R_{1,-}(z)J_{R,1}(z)+J_{R,2}(z), \qquad z \in \partial D(1, \varepsilon),
\end{align}
where
\begin{align}\label{jump:R2}
J_{R,2}(z) =\frac{9}{2 f(z)^{3}} \begin{pmatrix}
\mathfrak{u}_2& 0\\
0 & \mathfrak{v}_2
\end{pmatrix}
\end{align}
with $\mathfrak{u}_2$ and $\mathfrak{v}_2$ given in \eqref{def:ukvk}.
\item[\rm (c)] As $z \to \infty$, we have
$R_2(z) = \Boh(z^{-1})$.
\end{itemize}
\end{rhp}
From \eqref{def:R1}, \eqref{jump:R1}, \eqref{jump:R2} and Cauchy's residue theorem, it follows that
\begin{align}
R_2(z) = \frac{1}{2 \pi \ii} \oint_{\partial D(1, \varepsilon)}\frac{R_{1,-}(s)J_{R,1}(s)+J_{R,2}(s)}{z-s} \ud s
\end{align}
is a diagonal matrix. For general $k \geq 3$, the functions $R_k$ are analytic in $\mathbb{C} \setminus \partial D(1, \varepsilon)$ with asymptotic behavior $\Boh (1/z)$ as $z \to \infty$, and satisfy
\begin{align}
R_{k,+}(z) = R_{k,-}(z) + \sum_{l=1}^k R_{k-l,-}(z) J_{R,l}(z), \qquad z \in \partial D(1, \varepsilon),
\end{align}
where the functions $J_{R,k}(z)$ are given in \eqref{Jk}. By Cauchy's residue theorem, we have
\begin{equation}\label{70}
R_k(z) = \frac{1}{2 \pi \ii} \oint_{\partial D(1, \varepsilon)} \sum_{l=1}^k R_{k-l,-}(s) J_{R,l}(s) \frac{\ud s}{z-s}.
\end{equation}
One can check that, by the structure of $J_{R,k}(z)$ and mathematical induction, each $R_k$ takes the following structure:
\begin{equation}
R_k(z) = \begin{cases}
\begin{pmatrix}
0 & \left(R_k(z)\right)_{12}\\
\left(R_k(z)\right)_{21} & 0
\end{pmatrix}, & \textrm{for odd $k$,}
\\
\begin{pmatrix}
\left(R_k(z)\right)_{11} & 0\\
0 & \left(R_k(z)\right)_{22}
\end{pmatrix}, & \textrm{for even $k$.}
\end{cases}
\end{equation}
This, together with \eqref{eq:Rexp}, gives us
\begin{align}\label{R1}
h_{\nu}^{\frac 14 \sigma_3} R(z) h_{\nu}^{-\frac 14 \sigma_3} &\sim I + \sum_{k=0}^{\infty} \left(h_{\nu}^{3k+1} \begin{pmatrix}
0 & 0\\
\left(R_{2k+1}(z)\right)_{21} & 0
\end{pmatrix}
+h_{\nu}^{3k+2} \begin{pmatrix}
0 & \left(R_{2k+1}(z)\right)_{12}\\
0 & 0
\end{pmatrix}\right.\\
&~~~\left.+h_{\nu}^{3k+3} \begin{pmatrix}
\left(R_{2k+2}(z)\right)_{11} & 0\\
0 & \left(R_{2k+2}(z)\right)_{22}
\end{pmatrix} \right), \qquad \nu \to +\infty.\nonumber
\end{align}
We are now ready to prove our main result.

\section{Proof of Theorem \ref{th:1}}\label{sec:proof}
Recall the RH characterization of the Bessel kernel given in \eqref{def:bessel}, we then follow the series of transformations $\Psi \to Y \to T \to S$ in \eqref{def:Y}, \eqref{def:T} and \eqref{def:S} to obtain that
\begin{multline}
\nu^2 K_{\nu}^{\Bes}(\nu^2 u, \nu^2 v) = \frac{1}{2 \pi \ii (u-v)} \begin{pmatrix}
0 &1
\end{pmatrix}e^{\nu (g_+(v)-\pi \ii/2) \sigma_3}S_+(v)^{-1}S_+(u)
\\
\times e^{-\nu (g_+(u)-\pi \ii/2)  \sigma_3}\begin{pmatrix}
1\\0
\end{pmatrix}, \qquad u,v>0,
\end{multline}
where $g$ is given in \eqref{def:g}. In what follows, we split our discussions into different cases based on different ranges of $u$ and $v$.

If $u \in (0, 1- \varepsilon]$ and $v \in (1- \varepsilon, 1+ \varepsilon)$, applying the final transformation \eqref{def:R} and \eqref{def:P1} shows that
\begin{align}\label{xto1}
&\nu^2 K_{\nu}^{\Bes}(\nu^2 u, \nu^2 v) \\
&= \frac{1}{2 \pi \ii (u-v)} \begin{pmatrix}
0 &1
\end{pmatrix}e^{\nu (g_+(v)-\pi \ii/2) \sigma_3}P_+(v)^{-1}R(v)^{-1}R_+(u)N(u)e^{-\nu (g_+(u)-\pi \ii/2)  \sigma_3}\begin{pmatrix}
1\\0
\end{pmatrix}\nonumber\\
&=-\frac{e^{-\nu (g_+(u)-\pi \ii/2)}}{\ii \sqrt{2 \pi} (u-v)} \begin{pmatrix}
\ii \Ai'(\nu^{\frac 23}f(v)) &\Ai (\nu^{\frac 23}f(v))
\end{pmatrix}E(v)^{-1}R(v)^{-1}R_+(u)N(u)\begin{pmatrix}
1\\0
\end{pmatrix},\nonumber
\end{align}
where the functions $E,R$ and $N$ are given in \eqref{def:E},  \eqref{def:R} and \eqref{def:N}, respectively.
By \eqref{g01}, it is readily seen that
\begin{equation}\label{eq:g+deri}
g_+'(x) = -\frac{\sqrt{1-x}}{2x}<0, \qquad x \in (0, 1-\varepsilon].
\end{equation}
Thus, for $u \in (0, 1- \varepsilon]$,
\begin{equation}
e^{-\nu (g_+(u)-\pi \ii/2)} \le e^{-\nu (g_+(1-\varepsilon)-\pi \ii/2)} =  e^{-\nu\left(\frac 12 \ln \frac{1+\sqrt{\varepsilon}}{1-\sqrt{\varepsilon}}-\sqrt{\varepsilon}\right)}<e^{-\frac 32 h_{\nu}^{-1}}, \qquad \nu \to +\infty,
\end{equation}
and other terms in \eqref{xto1} are bounded for large positive $\nu$, which follow from \eqref{def:N}, \eqref{def:E} and \eqref{eq:Rexp}. By taking
\begin{equation}\label{uv}
u=(1-xh_{\nu})^2, \qquad v=(1-yh_{\nu})^2
\end{equation}
in \eqref{xto1}, where
\begin{equation}\label{eq:case1}
(1-\sqrt{1-\varepsilon}) h_{\nu}^{-1}\le x < h_{\nu}^{-1}, \qquad t_0 \le y < (1-\sqrt{1-\varepsilon}) h_{\nu}^{-1}
\end{equation}
with $t_0$ being any fixed real number, we have, for any $\mathfrak{m}\in\mathbb{N}$,
\begin{align}
&\sqrt{\phi_{\nu}'(x) \phi_{\nu}'(y)}K_{\nu}^{\Bes}(\phi_{\nu}(x), \phi_{\nu}(y))
<e^{-\frac 32 h_{\nu}^{-1}} \cdot \Boh (e^{-y})\\
&<e^{-\frac 12 h_{\nu}^{-1}} \cdot \Boh \left(e^{-(x+y)}\right) = h_{\nu}^{\mathfrak{m}+1}\cdot \Boh \left(e^{-(x+y)}\right),\qquad h_{\nu} \to 0^+.\nonumber
\end{align}
Here, $\phi_{\nu}$ is defined in \eqref{def:phi} and the error term $\Boh (e^{-y})$ in the first inequality comes from the estimate \cite[Formula 9.7.15]{DLMF}
\begin{equation}\label{estimate-Ai}
|p(\zeta)| \cdot \max \left(|\Ai (\zeta)|, |\Ai' (\zeta)|\right) \le c_p e^{-\zeta}, \qquad \zeta \in \mathbb{R},
\end{equation}
where $p$ is an arbitrary polynomial and the constant $c_p$ only depends on $p$. As a consequence, the transformed Bessel kernel \eqref{def:inducedkernel} will be absorbed into the error term of \eqref{h-s} in this case. As for the expansion terms in \eqref{h-s}, note that $x$ is large as $\nu \to +\infty$, we obtain again from \cite[Formula 9.7.15]{DLMF} that, for any arbitrary polynomials $q$,
\begin{equation}\label{Aix}
q(x) \cdot \Ai (x) \le q(x) \cdot \frac{e^{-\frac 23 x^{3/2}}}{2 \sqrt{\pi} x^{1/4}} < e^{-\frac 32 h_{\nu}^{-1}} < h_{\nu}^{\mathfrak{m}+1}\cdot \Boh \left(e^{-x}\right),
\end{equation}
and
\begin{equation}\label{Ai'x}
q(x) \cdot |\Ai' (x)| \le q(x) \cdot \frac{x^{1/4}e^{-\frac 23 x^{3/2}}}{2 \sqrt{\pi}} \left(1+\frac{7}{48 x^{3/2}}\right)< e^{-\frac 32 h_{\nu}^{-1}} < h_{\nu}^{\mathfrak{m}+1}\cdot \Boh \left(e^{-x}\right).
\end{equation}
This, together with the estimate \eqref{estimate-Ai} for $\Ai (y)$ and $\Ai'(y)$, implies that the expansion terms in \eqref{h-s} is also absorbed into the error term, which shows that \eqref{h-s} is valid under the condition \eqref{eq:case1}.

A similar argument holds if $u \in (1- \varepsilon, 1+ \varepsilon)$ and $v \in (0, 1- \varepsilon]$, which implies \eqref{h-s} for $t_0 \le x < (1-\sqrt{1-\varepsilon}) h_{\nu}^{-1}$ and $(1-\sqrt{1-\varepsilon}) h_{\nu}^{-1}\le y < h_{\nu}^{-1}$.

If $0<u,v\le 1- \varepsilon$, we obtain from \eqref{def:R} that
\begin{align}\label{yto1}
&\nu^2 K_{\nu}^{\Bes}(\nu^2 u, \nu^2 v)\\
&= \frac{1}{2 \pi \ii (u-v)} \begin{pmatrix}
0 &1
\end{pmatrix}e^{\nu (g_+(v)-\pi \ii/2) \sigma_3}N(v)^{-1}R_+(v)^{-1}R_+(u)N(u)e^{-\nu (g_+(u)-\pi \ii/2)  \sigma_3}\begin{pmatrix}
1\\0
\end{pmatrix}\nonumber\\
&=\frac{e^{-\nu (g_+(u)-\pi \ii/2)}e^{-\nu (g_+(v)-\pi \ii/2)}}{2 \pi \ii (u-v)}  \begin{pmatrix}
0 &1
\end{pmatrix}N(v)^{-1}R_+(v)^{-1}R_+(u)N(u)\begin{pmatrix}
1\\0
\end{pmatrix}.\nonumber
\end{align}
From \eqref{eq:g+deri}, it follows that
\begin{multline}
e^{-\nu (g_+(u)-\pi \ii/2)}e^{-\nu (g_+(v)-\pi \ii/2)} \le e^{-2\nu (g_+(1-\varepsilon)-\pi \ii/2)}
\\
=  e^{-\nu\left(\ln \frac{1+\sqrt{\varepsilon}}{1-\sqrt{\varepsilon}}-2\sqrt{\varepsilon}\right)}<e^{-3 h_{\nu}^{-1}},  \qquad \nu \to +\infty,
\end{multline}
and other terms in \eqref{yto1} are bounded for large positive $\nu$, as can be seen from their definitions in \eqref{def:N} and \eqref{eq:Rexp}. Substituting \eqref{uv} into \eqref{yto1} with
\begin{equation}\label{eq:case3}
(1-\sqrt{1-\varepsilon}) h_{\nu}^{-1}\le x, y < h_{\nu}^{-1},
\end{equation}
we have, for any $\mathfrak{m}\in\mathbb{N}$,
\begin{align}
&\sqrt{\phi_{\nu}'(x) \phi_{\nu}'(y)}K_{\nu}^{\Bes}(\phi_{\nu}(x), \phi_{\nu}(y)) <e^{-3 h_{\nu}^{-1}} \cdot \Boh (1)\\
&<e^{-  h_{\nu}^{-1}} \cdot \Boh \left(e^{-(x+y)}\right) = h_{\nu}^{\mathfrak{m}+1}\cdot \Boh \left(e^{-(x+y)}\right), \qquad h_{\nu} \to 0^+.\nonumber
\end{align}
The transformed Bessel kernel is again absorbed into the error term completely in this case. The removability of the singularities at $x=y$ comes from the symmetric structure of the expansion \eqref{yto1}. Since the Airy functions in the expansion \eqref{h-s} related to $x$ and $y$ are superexponential decay as $\nu \to +\infty$ (see \eqref{Aix} and \eqref{Ai'x}), we conclude \eqref{h-s} under the condition \eqref{eq:case3}.

It remains to consider the final case, namely, $1-\varepsilon <u,v< 1+ \varepsilon$, which corresponds to
\begin{equation}\label{eq:case4}
t_0 \le x, y< (1-\sqrt{1-\varepsilon}) h_{\nu}^{-1}
\end{equation}
through \eqref{uv}. To proceed, we again observe from
\eqref{def:R} and \eqref{def:P1} that
\begin{align}\label{xyto1}
&\nu^2 K_{\nu}^{\Bes}(\nu^2 u, \nu^2 v)\\
&= \frac{1}{2 \pi \ii (u-v)} \begin{pmatrix}
0 &1
\end{pmatrix}e^{\nu (g_+(v)-\pi \ii/2) \sigma_3}P_+(v)^{-1}R(v)^{-1}R(u)P_+(u)e^{-\nu (g_+(u)-\pi \ii/2)  \sigma_3}\begin{pmatrix}
1\\0
\end{pmatrix}\nonumber\\
&=-\frac{1}{\ii (u-v)} \begin{pmatrix}
\ii \Ai'(\nu^{\frac 23}f(v)) &\Ai (\nu^{\frac 23}f(v))
\end{pmatrix}E(v)^{-1}R(v)^{-1}R(u)E(u)\begin{pmatrix}
\Ai(\nu^{\frac 23}f(u)) \\-\ii \Ai'(\nu^{\frac 23}f(u))
\end{pmatrix}.\nonumber
\end{align}
Inserting \eqref{uv} into the above formula, we see from \eqref{def:inducedkernel} that
\begin{align}\label{89}
&\hat K_{\nu}^{\mathrm{Bes}}(x,y)=\sqrt{\phi_{\nu}'(x) \phi_{\nu}'(y)}K_{\nu}^{\Bes}(\phi_{\nu}(x), \phi_{\nu}(y))\nonumber\\
&= \frac{\sqrt{(1-h_{\nu}x)(1-h_{\nu}y)}}{x-y-\frac{h_{\nu}}{2}(x^2-y^2)}\begin{pmatrix}
 \Ai'(\nu^{\frac 23}f((1-h_{\nu}y)^2)) &-\ii \Ai (\nu^{\frac 23}f((1-h_{\nu}y)^2))
\end{pmatrix}\nonumber\\
&\quad \times E((1-h_{\nu}y)^2)^{-1}R((1-h_{\nu}y)^2)^{-1}R((1-h_{\nu}x)^2)E((1-h_{\nu}x)^2) \nonumber \\
&\quad \times \begin{pmatrix}
\Ai(\nu^{\frac 23}f((1-h_{\nu}x)^2)) \\-\ii \Ai'(\nu^{\frac 23}f((1-h_{\nu}x)^2))
\end{pmatrix}.
\end{align}

We now show expansions of different parts on the right-hand side of the above formula. According to \cite[Lemma 3.2]{B23}, one has
\begin{equation}\label{final:1}
\frac{\sqrt{(1-h_{\nu}x)(1-h_{\nu}y)}}{x-y-\frac{h_{\nu}}{2}(x^2-y^2)} = \frac{1}{x-y} - (x-y)\sum_{j=2}^{\infty}r_j(x,y)h_{\nu}^j,
\end{equation}
where each $r_j(x,y)$ is certain polynomial of degree $j-2$.

Next, it is observed that if $x=y$,
\begin{equation}
E((1-h_{\nu}y)^2)^{-1}R((1-h_{\nu}y)^2)^{-1}R((1-h_{\nu}x)^2)E((1-h_{\nu}x)^2)= I.
\end{equation}
This, together with \eqref{R1} and the fact that $E(z)$ is an analytic function in $D(1, \varepsilon)$, implies that,
as $\nu \to +\infty$,
\begin{multline}\label{final:2}
E((1-h_{\nu}y)^2)^{-1}R((1-h_{\nu}y)^2)^{-1}R((1-h_{\nu}x)^2)E((1-h_{\nu}x)^2)
\\
\sim I + (x-y)\sum_{j=1}^{\infty}e_j(x,y) h_{\nu}^j,
\end{multline}
where $e_j(x,y)$ are certain matrices with all the entries being polynomials in $x$ and $y$. Indeed, by \eqref{E1}, \eqref{def:R1} and \eqref{R1}, it follows that, as $h_\nu \to 0^+$,
\begin{align}\label{93}
E((1-h_{\nu}x)^2)&=(2h_{\nu})^{-\frac 14 \sigma_3}e^{-\frac{\pi \ii}{4} \sigma_3} \sigma_3 \left(I + \frac{x}{5}\sigma_3h_{\nu} + \begin{pmatrix}
\frac{3}{35}x^2 & 0\\
0 & -\frac{8}{175}x^2
\end{pmatrix}h_{\nu}^2 + \Boh(h_{\nu}^3) \right),\\
R((1-h_{\nu}x)^2)&=h_{\nu}^{-\frac 14 \sigma_3} \left(I + \begin{pmatrix}
0 & 0\\
\frac{7\sqrt{2}}{40} &0
\end{pmatrix}h_{\nu}+\begin{pmatrix}
0&-\frac{\sqrt{2}}{70}\\
\frac{\sqrt{2}}{25}x&0
\end{pmatrix} h_{\nu}^2 + \Boh(h_{\nu}^3)\right)h_{\nu}^{\frac 14 \sigma_3}.\label{94}
\end{align}
Thus, we have
\begin{align}
e_1(x,y) = \frac{1}{5} \sigma_3, \qquad
e_2(x,y)  = \begin{pmatrix}
\frac{15x+8y}{175} & 0\\
\frac{\ii}{25}&-\frac{8x+15y}{175}
\end{pmatrix}.
\end{align}

To deal with the parts involving the Airy functions, we state the following proposition. 
\begin{proposition}\label{pro:Ai}
Let $\xi$ and $\eta$ be two variables with expansions
\begin{equation}\label{def:xi}
\xi = x + \sum_{j=1}^{\infty} p_{1,j}(x)h^j, \qquad \eta = y + \sum_{j=1}^{\infty} p_{1,j}(y)h^j, \qquad h\to 0,
\end{equation}
where $p_j$ are polynomials of degree $j+1$. As $h \to 0$, we have 
\begin{align}\label{expand-Ai}
& \frac{1}{x-y}\begin{pmatrix}
\Ai'(\eta) &
-\ii \Ai(\eta)
\end{pmatrix} \begin{pmatrix}
\Ai(\xi) \\
-\ii \Ai'(\xi)
\end{pmatrix}
 \\
& = K^{\Ai}(x,y) + \sum_{j=1}^{\infty} \left(a_j(x,y)\Ai(x)\Ai(y)+b_j(x,y)\Ai(x)\Ai'(y)+b_j(y,x)\Ai'(x)\Ai(y)\right.\nonumber\\
& \quad \left.+c_j(x,y)\Ai'(x)\Ai'(y)\right)h^j, \nonumber 
\end{align}
where $K^{\Ai}(x,y)$ denotes the Airy kernel given in \eqref{def:KAi}, $a_j(x,y)$, $b_j(x,y)$ and $c_j(x,y)$ represent certain polynomials in $x$ and $y$.
\end{proposition}
\begin{proof}
By noting
\begin{align}\label{fm}
\frac{1}{m!}\left(\xi-x\right)^m =\frac{1}{m!}\left(\sum_{j=1}^{\infty} p_{1,j}(x)h^j\right)^m = \sum_{j=m}^{\infty} p_{m,j}(x) h^j, \qquad h \to 0,
\end{align}
it is easily seen that
\begin{align}\label{sum-mn}
p_{m,j}(x) =  \frac{(m-n)! n!}{m!} \sum_{k=m-n}^{j-n}  p_{m-n,k}(x)p_{n,j-k}(x)
\end{align}
for $1<n<m\le j$.
We then obtain from the analyticity of Airy function and \eqref{fm} that
\begin{align}\label{exp:Ai}
\Ai (\xi)= \Ai (x) + \sum_{j=1}^{\infty} \sum_{m=1}^j p_{m,j}(x) \Ai^{(m)}(x)h^j, \qquad h \to 0,
\end{align}
where $\Ai^{(m)}(x)$ denotes the $m$-th derivative of $\Ai(x)$ with respect to $x$. From the differential equation
\begin{equation}
\frac{\ud^2 \Ai(z)}{\ud z^2} = z \Ai (z)
\end{equation}
satisfied by Airy function, a direct calculation gives us
\begin{align}\label{Aim}
\Ai^{(m)}(x) = P_{m}(x) \Ai(x) + Q_m(x) \Ai'(x),
\end{align}
where $P_m (x)$ and $Q_m (x)$ are polynomials with $P_0(x)=1$ and $Q_0(x)=0$. They satisfy the recurrence relations (cf. \cite{L11})
\begin{align}\label{3termP}
P_{n+1} (x) = P_n'(x)+x Q_n(x),\qquad
Q_{n+1} (x) = Q_n'(x)+P_n(x).
\end{align}
With the aids of the expansions \eqref{exp:Ai} and \eqref{Aim}, we obtain
\begin{align}\label{airypart}
&\quad\begin{pmatrix}
 \Ai'(\eta) &-\ii \Ai (\eta)
\end{pmatrix}\begin{pmatrix}
\Ai(\xi) \\-\ii \Ai'(\xi)
\end{pmatrix}\\
&=\left(\Ai (x) + \sum_{j=1}^{\infty} \sum_{m=1}^j p_{m,j}(x) \Ai^{(m)}(x)h^j\right) \cdot \left(\Ai' (y) + \sum_{j=1}^{\infty} \sum_{m=1}^j p_{m,j}(y) \Ai^{(m+1)}(y)h^j\right)\nonumber\\
&\quad - \left(\Ai' (x) + \sum_{j=1}^{\infty} \sum_{m=1}^j p_{m,j}(x) \Ai^{(m+1)}(x)h^j\right) \cdot \left(\Ai (y) + \sum_{j=1}^{\infty} \sum_{m=1}^j p_{m,j}(y) \Ai^{(m)}(y)h^j\right)\nonumber\\
&=\Ai (x) \Ai'(y)-\Ai'(x)\Ai(y) + \sum_{N=1}^{\infty}\left(a_{N,00}(x,y) \Ai(x) \Ai(y)+a_{N,01}(x,y) \Ai(x) \Ai'(y)\right.\nonumber\\
&\quad \left.-a_{N,01}(y,x) \Ai'(x) \Ai(y)+a_{N,11}(x,y) \Ai'(x) \Ai'(y)\right)h^N,\nonumber
\end{align}
where
\begin{align*}
&a_{N,00}(x,y) = \sum_{n=1}^{N} \left(p_{n,N}(y) P_{n+1}(y) -p_{n,N}(x) P_{n+1}(x)\right)\\
&\quad+ \sum_{\substack{j+k=N\\j\ge1,k\ge1}} \sum_{m=1}^{j} \sum_{n=1}^k \left(p_{m,j}(x) p_{n,k}(y) P_m(x) P_{n+1}(y)-p_{m,j}(y) p_{n,k}(x) P_m(y) P_{n+1}(x)\right),\nonumber\\
&a_{N,01}(x,y) = \sum_{n=2}^{N} \left(p_{n,N}(y) Q_{n+1}(y) +p_{n,N}(x) P_{n}(x)\right)\\
&\quad+ \sum_{\substack{j+k=N\\j\ge1,k\ge1}} \sum_{m=1}^{j} \sum_{n=1}^k p_{m,j}(x) p_{n,k}(y)\left( P_m(x) Q_{n+1}(y)- P_{m+1}(x) Q_{n}(y)\right),\nonumber
\end{align*}
and
\begin{align}
&a_{N,11}(x,y) = \sum_{n=1}^{N} \left(p_{n,N}(x) Q_{n}(x) -p_{n,N}(y) Q_{n}(y)\right)\nonumber\\
&\quad+ \sum_{\substack{j+k=N\\j\ge1,k\ge1}} \sum_{m=1}^{j} \sum_{n=1}^k \left(p_{m,j}(x) p_{n,k}(y) Q_m(x) Q_{n+1}(y)-p_{m,j}(y) p_{n,k}(x) Q_m(y) Q_{n+1}(x)\right),\nonumber
\end{align}
are polynomials in $x$ and $y$. Since the polynomials $a_{N,00}(x,y)$ and $a_{N,11}(x,y)$ are anti-symmetric in $x$ and $y$, they must have the form
\begin{equation}\label{xy}
(x-y) \times (\textrm{polynomials in $x$ and $y$}).
\end{equation}
We would like to show the polynomials $a_{N,01}(x,y)$ admit the same structure. In other words, we want to show
\begin{equation}
a_{N,01}(x,x) = 0,
\end{equation}
or equivalently,
\begin{align}\label{p01}
&\sum_{n=2}^N p_{n,N}(x) \left(Q_{n+1} (x)+ P_n(x)\right) \nonumber\\
&+ \sum_{\substack{j+k=N\\j\ge 1, k\ge 1}}^{N} \sum_{m=1}^j \sum_{n=1}^k \left(p_{m,j}(x) p_{n,k}(x) (P_m (x) Q_{n+1} (x)-P_{m+1}(x) Q_n (x)\right)=0.
\end{align}
To prove the above equality, we need the following lemma.
\begin{lemma}
With polynomials $P_m$ and $Q_m$ defined through \eqref{Aim}, we have
\begin{align}\label{eq:1}
\sum_{j=0}^N \frac{1}{j! (N-j)!} \left(P_j(x) Q_{N+1-j}(x) - Q_j(x) P_{N+1-j}(x)\right)=0, \qquad N\geq 1.
\end{align}
\end{lemma}
\begin{proof}
We use the method of mathematical induction to prove the above identity. It is clear that \eqref{eq:1} holds for $N=1$.
Assume that \eqref{eq:1} is valid for $N=k >1$, i.e.,
\begin{align}
\sum_{j=0}^k \frac{1}{j! (k-j)!} \left(P_j(x) Q_{k+1-j}(x) - Q_j(x) P_{k+1-j}(x)\right)=0.
\end{align}
After taking derivative on both sides with respected to $x$, it follows that
\begin{multline*}
\sum_{j=0}^k \frac{1}{j! (k-j)!} \left(P_j'(x) Q_{k+1-j}(x)+ P_j(x) Q_{k+1-j}'(x)- Q_j'(x) P_{k+1-j}(x)\right.
\\
\left.-Q_j(x) P_{k+1-j}'(x)\right)=0,
\end{multline*}
Applying the recurrence relations \eqref{3termP} to the above formula, we arrive at
\begin{align}
&\sum_{j=0}^k \frac{1}{j! (k-j)!} \left(P_{j+1}(x) Q_{k+1-j}(x)- Q_{j+1}(x) P_{k+1-j}(x)+P_j(x) Q_{k+2-j}(x)\right.\nonumber\\
&\qquad\left.-Q_j(x) P_{k+2-j}(x)\right)\nonumber\\
&=\sum_{j=0}^{k+1} \frac{j}{j! (k+1-j)!} \left(P_{j}(x) Q_{k+2-j}(x)- Q_{j}(x) P_{k+2-j}(x)\right)\nonumber\\
&\qquad+\sum_{j=0}^{k+1}\frac{k+1-j}{j! (k+1-j)!}\left(P_j(x) Q_{k+2-j}(x)-Q_j(x) P_{k+2-j}(x)\right)\nonumber\\
&=\sum_{j=0}^{k+1} \frac{k+1}{j! (k+1-j)!} \left(P_j(x) Q_{k+2-j}(x) - Q_j(x) P_{k+2-j}(x)\right)\nonumber=0,
\end{align}
which is \eqref{eq:1} with $N=k+1$.
\end{proof}
Using \eqref{sum-mn} and \eqref{eq:1}, we could rewrite the left-hand side of \eqref{p01} as
\begin{align*}
&\sum_{n=2}^N p_{n,N}(x) \left(Q_{n+1} (x)+ P_n(x)\right) \\
&\qquad+ \sum_{\substack{m+n=2\\m\ge 1, n\ge 1}}^{N} \sum_{\substack{j+k=N\\ j \ge m, k \ge n}}\left(p_{m,j}(x) p_{n,k}(x) (P_m (x) Q_{n+1} (x)-P_{m+1}(x) Q_n (x))\right)\nonumber\\
&=\sum_{t=2}^N \Biggl(p_{t,N}(x) \left(Q_{t+1} (x)+ P_t(x)\right)\\
&\qquad + \sum_{\substack{m+n=t\\m\ge 1, n\ge 1}} \sum_{\substack{j+k=N,\\j \ge m, k \ge n}}\left(p_{m,j}(x) p_{n,k}(x) (P_m (x) Q_{n+1} (x)-P_{m+1}(x) Q_n (x))\right)\Biggr)\nonumber\\
&=\sum_{t=2}^N \Biggl(p_{t,N}(x) \left(Q_{t+1} (x)+ P_t(x)\right)\\
&\qquad + \sum_{\substack{m+n=t\\m\ge 1, n\ge 1}} \left(\frac{t!}{m! n!} p_{t,N}(x) (P_m (x) Q_{n+1} (x)-P_{m+1}(x) Q_n (x))\right)\Biggr)\\
&=\sum_{t=2}^N p_{t,N}(x) \Biggl(Q_{t+1} (x)+ P_t(x) \\
&\qquad+ t! \sum_{m= 1}^{t-1} \left(\frac{1}{m! (t-m)!} (P_m (x) Q_{t+1-m} (x)-P_{t+1-m}(x) Q_m (x))\right)\Biggr)\nonumber\\
&=\sum_{t=2}^N t! p_{t,N}(x) \left(\sum_{m= 0}^{t} \left(\frac{1}{m! (t-m)!} (P_m (x) Q_{t+1-m} (x)-P_{t+1-m}(x) Q_m (x))\right)\right)=0,
\end{align*}
as required. 

Since the polynomials $a_{N,00}, a_{N,01}$ and $a_{N,11}$ all take the same structure \eqref{xy}, we obtain \eqref{expand-Ai} in Proposition \ref{pro:Ai} from  \eqref{airypart}.
\end{proof}
From the expansion of $f$ near $z=1$ given in \eqref{def:f}, we have
\begin{align}\label{110}
\nu^{\frac 23}f((1-h_{\nu}x)^2) = x + \sum_{j=1}^{\infty} \tilde p_{1,j}(x) h_{\nu}^j, \qquad \textrm{$\nu \to +\infty$},
\end{align}
where $\tilde p_{1,j}(x)$ are certain polynomials of degree $j+1$ with
\begin{align}
\tilde p_{1,1}(x) = \frac{3}{10} x^2,\qquad
\tilde p_{1,2}(x)  = \frac{32}{175} x^3.
\end{align}
Then by setting
\begin{equation}\label{def:xieta}
\xi = \nu^{\frac 23}f((1-h_{\nu}x)^2), \qquad \eta =\nu^{\frac 23}f((1-h_{\nu}y)^2),
\end{equation}
applying Proposition \ref{pro:Ai}, and employing the estimates of Airy functions in \eqref{estimate-Ai}, it follows that, for any $\mathfrak{m}\in\mathbb{N}$,
\begin{align}\label{final:3}
&\frac{1}{x-y}\begin{pmatrix}
 \Ai'(\nu^{\frac 23}f((1-h_{\nu}y)^2)) &-\ii \Ai (\nu^{\frac 23}f((1-h_{\nu}y)^2))
\end{pmatrix}\begin{pmatrix}
\Ai(\nu^{\frac 23}f((1-h_{\nu}x)^2)) \\-\ii \Ai'(\nu^{\frac 23}f((1-h_{\nu}x)^2))
\end{pmatrix}\\
&=\frac{\Ai (x) \Ai'(y)-\Ai'(x)\Ai(y)}{x-y} + \sum_{N=1}^{\mathfrak{m}}\left(\tilde a_{N,00}(x,y) \Ai(x) \Ai(y)+\tilde a_{N,01}(x,y) \Ai(x) \Ai'(y)\right.\nonumber\\
&\quad \left.+\tilde a_{N,01}(y,x) \Ai'(x) \Ai(y)+\tilde a_{N,11}(x,y) \Ai'(x) \Ai'(y)\right)h_{\nu}^N+h_{\nu}^{\mathfrak{m}+1} \cdot \Boh\left(e^{-(x+y)}\right),\nonumber
\end{align}
where $\tilde a_{N,00}(x,y)$, $\tilde a_{N,01}(x,y)$ and $\tilde a_{N,11}(x,y)$ are certain polynomials in $x$ and $y$.

Combining \eqref{final:1}, \eqref{final:2} and \eqref{final:3} together gives us that, under the condition \eqref{eq:case4},
\begin{align}
&\hat K_{\nu}^{\mathrm{Bes}}(x,y)=\sqrt{\phi_{\nu}'(x) \phi_{\nu}'(y)}K_{\nu}^{\Bes}(\phi_{\nu}(x), \phi_{\nu}(y)) \nonumber\\
&=\frac{\Ai (x) \Ai'(y)-\Ai'(x)\Ai(y)}{x-y} + \sum_{j=1}^{\mathfrak{m}}\left(p_{j,00}(x,y) \Ai(x) \Ai(y)+p_{j,01}(x,y) \Ai(x) \Ai'(y)\right.\nonumber\\
&\quad \left.+p_{j,10}(x,y) \Ai'(x) \Ai(y)+p_{j,11}(x,y) \Ai'(x) \Ai'(y)\right)h_{\nu}^j+h_{\nu}^{\mathfrak{m}+1} \cdot \Boh\left(e^{-(x+y)}\right),\nonumber
\end{align}
as $h_{\nu} \to 0^+$, where $p_{j, \kappa \lambda} (x,y)$, $\kappa, \lambda \in \{0,1\}$, are polynomials in $x$ and $y$.
Additionally, it is worth noting that even though the calculations involved are inherently complicated, it is possible to derive precise formulae for $p_{j, \kappa \lambda}$ utilizing the explicit expressions \eqref{def:f}, \eqref{E1}, \eqref{R1} and \eqref{final:1}; see Remark \ref{rk:coeffcal} below for more information. The first few polynomials are hereby presented:
\begin{align*}
p_{1,00}(x,y)&=-\frac{3}{10}(x^2+xy+y^2),
\\
p_{1,01}(x,y)&=p_{1,10}(x,y)=\frac 15,
\\
p_{1,11}(x,y)&=\frac{3}{10}(x+y),
\end{align*}
and
\begin{align*}
p_{2,00}(x,y)&=\frac{56-235(x^2+y^2)-319xy(x+y)}{1400},
\\
p_{2,01}(x,y)&=p_{2,10}(y,x)=\frac{63(x^4+x^3y-x^2y^2-xy^3-y^4)-55x+239y}{1400},
\\
p_{2,11}(x,y)&=\frac{340(x^2+y^2)+256xy}{1400},
\end{align*}
which lead to \eqref{def:K1} and \eqref{def:K2}.

Finally, since all the terms $\Ai (\cdot)$, $\Ai' (\cdot)$, $E (\cdot)$ and $R (\cdot)$ are analytic functions for $t_0 \le x, y < (1-\sqrt{1-\varepsilon}) h_{\nu}^{-1}$, all the expansions can be repeatedly differentiated with respect to the variables $x$ and $y$ while preserving the uniformity, which holds for the kernel expansion.

This completes the proof of Theorem \ref{th:1}.
\qed

\begin{remark}
It is worthwhile to see that \cite[Lemma 3.5]{B23} follows directly from the above proof. Indeed, with $\xi$ and $\eta$ defined in \eqref{def:xieta}, we have from \eqref{110} that 
    \begin{align}
        \xi &= \nu^{\frac 23}f((1-h_{\nu}x)^2) = x + \frac{3x^2}{10}h_{\nu}+\frac{32x^3}{175}h_{\nu}^2+\frac{1037x^4}{7875}h_{\nu}^3+\cdots, \qquad \textrm{$h_{\nu} \to 0$},
    \\
        \eta & = \nu^{\frac 23}f((1-h_{\nu}y)^2)= y + \frac{3y^2}{10}h_{\nu}+\frac{32y^3}{175}h_{\nu}^2+\frac{1037y^4}{7875}h_{\nu}^3+\cdots, \qquad \textrm{$h_{\nu} \to 0$}.
    \end{align}
Reversing the above series gives
    \begin{align}
        x &= \xi - \frac{3\xi^2}{10}h_{\nu}-\frac{\xi^3}{350}h_{\nu}^2+\frac{479\xi^4}{63000}h_{\nu}^3+\cdots,\label{x-xi}\\
        y &= \eta - \frac{3\eta^2}{10}h_{\nu}-\frac{\eta^3}{350}h_{\nu}^2+\frac{479\eta^4}{63000}h_{\nu}^3+\cdots.\label{y-eta}
    \end{align}
Thus, rewriting the Bessel kernel using the variables $\xi$ and $\eta$ allows us to obtain the factorization by substituting \eqref{x-xi} and \eqref{y-eta} into \eqref{89}, i.e., 
    \begin{align}
        \tilde K_{\nu}^{\Bes}(\xi, \eta) = \frac{1}{\xi-\eta}\tilde T_0(\xi,\eta)\begin{pmatrix}
            \Ai'(\eta) & -\ii \Ai(\eta)
        \end{pmatrix}\tilde T_1(\xi,\eta)\begin{pmatrix}
            \Ai(\xi) \\ -\ii \Ai'(\xi)
        \end{pmatrix},
    \end{align}
    where 
    \begin{align}
        \tilde T_0(\xi,\eta) \sim 1-(\xi-\eta)^2\sum_{j=2}^{\infty}\pi_j(\xi,\eta)h_{\nu}^j
    \end{align}
    with $\pi_j(\xi,\eta)$ being certain polynomials of degree $j-2$ in $\xi$ and $\eta$, and the first two of them are given by
    \begin{align}
        \pi_2(\xi,\eta)=\frac{6}{35}, \qquad \pi_3(\xi,\eta)=\frac{16}{225}(\xi+\eta).
    \end{align}
    Additionally, we have
    \begin{align}
        \tilde T_1(\xi,\eta) \sim I + (\xi-\eta)\sum_{j=1}^{\infty}\rho_j(\xi,\eta)h_{\nu}^j,
    \end{align}
    where $\rho_j(\xi,\eta)$ are certain matrices with all the entries being polynomials in $\xi$ and $\eta$, and the first two of them are given by
    \begin{align}
        \rho_1(\xi,\eta)=\frac{1}{5}\sigma_3, \qquad \rho_2(\xi,\eta)=\begin{pmatrix}
            \frac{9\xi-5\eta}{350} & 0\\
            \frac{\ii}{25} & \frac{5\xi-9\eta}{350}
        \end{pmatrix}.
    \end{align}
As a consequence, we immediately obtain a full expansion of $\tilde K_{\nu}^{\Bes}$ in terms of powers of $h_{\nu}$. By truncating the expansion to $\Boh(h_{\nu}^3)$, it reads  
 \begin{align}
        \tilde K_{\nu}^{\Bes}(\xi, \eta) = K^{\Ai}(\xi, \eta)+\tilde K_1(\xi, \eta)h_{\nu}+\tilde K_2(\xi, \eta)h_{\nu}^2+h_{\nu}^3\cdot\Boh\left(e^{-(\xi+ \eta)}\right),
        \end{align}
        where $K^{\Ai}$ denotes the Airy kernel given in \eqref{def:KAi},
        \begin{align}
            \tilde K_1(\xi, \eta) &= \frac{1}{5}\left(\Ai(\xi)\Ai'(\eta)+\Ai'(\xi)\Ai(\eta)\right),\\
            \tilde K_2(\xi, \eta) &= \frac{1}{350}\left(14\Ai(\xi)\Ai(\eta)+(-51\xi+55\eta)\Ai(\xi)\Ai'(\eta)+(55\xi-51\eta)\Ai'(\xi)\Ai(\eta)\right),
        \end{align}   
which is \cite[Lemma 3.5]{B23}. Here, we do not need to expand the Airy functions, which reduces the complexity of the original proof presented in \cite{B23}.                
   
\end{remark}

\begin{remark}\label{rk:coeffcal}
We emphasize that our approach also provides a systematic way to calculate the polynomial coefficients $p_{j, \kappa \lambda}$ of the expansion kernels $K_j(x,y)$ in \eqref{def:Kj}. 
In view of the factorization in \eqref{89}, we achieve this by computing appropriately truncated Laurent series of the functions $f(z)$, $E(z)$ and $R(z)$ at $z=1$. By utilizing \eqref{def:g} and \eqref{def:f} for the function $f(z)$, we can obtain the precise expressions for $\tilde p_{1,j}$ in \eqref{110}. Similarly, using \eqref{def:N} and \eqref{def:E} for the function $E(z)$,  along with \eqref{Jk}, \eqref{eq:Rexp}, \eqref{70} and \eqref{R1} for the function $R(z)$, we can compute precise expressions for $e_j(x,y)$ in \eqref{final:2}. All of these calculations are straightforward by using the truncated Laurent series except the Cauchy integral in \eqref{70}, which can be evaluated by Cauchy's residue theorem as
\begin{align}
R_k(z) = (\textrm{principal part of the Laurent series of $F(z)$ at $z=1$}) - F(z)
\end{align}
with $F(z) = \sum_{l=1}^k R_{l-k}(z) J_l(z)$ and $R_0(z)=I$.

Following these computations, one can obtain the polynomials factors of the kernels $K_j(x.y)$ explicitly, and the first ten factors agree with those reported in the Mathematica supplement of \cite{B23}.
\end{remark}

\begin{appendix}
\section{The Airy parametrix}\label{airy}
The Airy parametrix $\Phi^{({\Ai})}$ is the unique solution of the following RH problem.
\begin{rhp}
\hfill
\begin{itemize}
\item[\rm(a)] $\Phi^{(\mathrm{Ai})}(z)$ is analytic in $\mathbb{C} \setminus \{\cup_{j=1}^4 \Sigma_j \cup \{0\}\}$, where the contours $\Sigma_j$, $j=1, 2, 3, 4$, are indicated in Figure \ref{fig:Airy}.
\item[\rm(b)] $\Phi^{(\Ai)}(z)$ satisfies the jump condition
\begin{align}\label{jump:Airy}
\Phi^{({\Ai})}_+(z)=\Phi^{({\Ai})}_-(z)\begin{cases}
\begin{pmatrix}
1 & 1\\
0 & 1
\end{pmatrix}, & \qquad z \in \Sigma_1,\\
\begin{pmatrix}
1 & 0\\
1& 1
\end{pmatrix}, & \qquad z \in \Sigma_2 \cup \Sigma_4,\\
\begin{pmatrix}
0 & 1\\
-1 & 0
\end{pmatrix}, & \qquad z \in \Sigma_3.
\end{cases}
\end{align}
\item[\rm(c)] As $z \to \infty$, we have
\begin{align}\label{infty:Ai}
\Phi^{({\Ai})}(z) = \frac{1}{\sqrt{2}} \begin{pmatrix}
z^{-\frac 14} & 0\\
0 & z^{\frac 14}
\end{pmatrix} \begin{pmatrix}
1 & \ii\\
\ii & 1
\end{pmatrix}\left(I + \Boh(z^{-\frac 32})\right)e^{-\frac 23 z^{3/2} \sigma_3}.
\end{align}
\item[\rm(d)] $\Phi^{({\Ai})}(z)$ is bounded near the origin.
\end{itemize}
\end{rhp}

\begin{figure}[ht]
\begin{center}

\tikzset{every picture/.style={line width=0.75pt}} 

\begin{tikzpicture}[x=0.75pt,y=0.75pt,yscale=-1,xscale=1]

\draw    (52,120) -- (203,120) -- (333,120) ;
\draw [shift={(132.5,120)}, rotate = 180] [fill={rgb, 255:red, 0; green, 0; blue, 0 }  ][line width=0.08]  [draw opacity=0] (8.93,-4.29) -- (0,0) -- (8.93,4.29) -- cycle    ;
\draw [shift={(273,120)}, rotate = 180] [fill={rgb, 255:red, 0; green, 0; blue, 0 }  ][line width=0.08]  [draw opacity=0] (8.93,-4.29) -- (0,0) -- (8.93,4.29) -- cycle    ;
\draw    (116.5,20) -- (216.5,120) ;
\draw [shift={(170.04,73.54)}, rotate = 225] [fill={rgb, 255:red, 0; green, 0; blue, 0 }  ][line width=0.08]  [draw opacity=0] (8.93,-4.29) -- (0,0) -- (8.93,4.29) -- cycle    ;
\draw    (117.5,208) -- (216.5,120) ;
\draw [shift={(170.74,160.68)}, rotate = 138.37] [fill={rgb, 255:red, 0; green, 0; blue, 0 }  ][line width=0.08]  [draw opacity=0] (8.93,-4.29) -- (0,0) -- (8.93,4.29) -- cycle    ;
\draw  [fill={rgb, 255:red, 0; green, 0; blue, 0 }  ,fill opacity=1 ] (216.67,119.83) .. controls (216.67,119.19) and (216.15,118.66) .. (215.5,118.66) .. controls (214.85,118.66) and (214.33,119.19) .. (214.33,119.83) .. controls (214.33,120.48) and (214.85,121) .. (215.5,121) .. controls (216.15,121) and (216.67,120.48) .. (216.67,119.83) -- cycle ;

\draw (218.5,123) node [anchor=north west][inner sep=0.75pt]   [align=left] {$0$};
\draw (343,117) node [anchor=north west][inner sep=0.75pt]   [align=left] {$\Sigma_1$};
\draw (98,14) node [anchor=north west][inner sep=0.75pt]   [align=left] {$\Sigma_2$};
\draw (33,113) node [anchor=north west][inner sep=0.75pt]   [align=left] {$\Sigma_3$};
\draw (95,205) node [anchor=north west][inner sep=0.75pt]   [align=left] {$\Sigma_4$};

\end{tikzpicture}
  \caption{The jump contours of the RH problem for $\Phi^{({\Ai})}$.}
   \label{fig:Airy}
\end{center}
\end{figure}
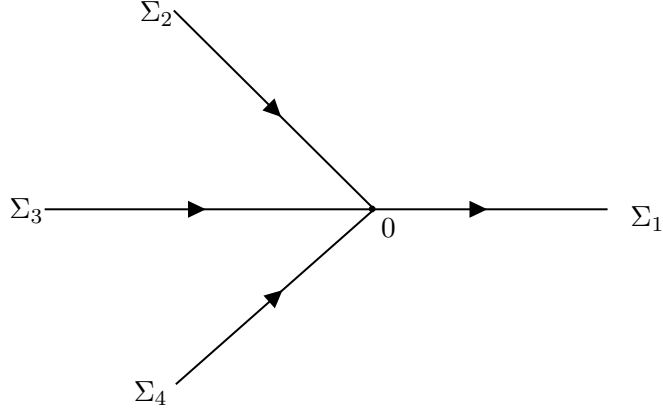

Denote $\omega:= e^{ 2 \pi \ii/3}$, the unique solution is given by (cf. \cite{Deift1999})
\begin{align}
\Phi^{({\Ai})}(z) = \sqrt{2 \pi} \begin{cases}
\begin{pmatrix}
\Ai (z) & - \omega^2 \Ai (\omega^2 z)\\
-\ii \Ai'(z) & \ii \omega \Ai' (\omega^2 z)
\end{pmatrix}, & \arg z \in \left(0, \frac{3 \pi}{4} \right),\\
\begin{pmatrix}
-\omega \Ai (\omega z) & - \omega^2 \Ai (\omega^2 z)\\
\ii \omega^2 \Ai'(\omega z) & \ii \omega \Ai' (\omega^2 z)
\end{pmatrix}, & \arg z \in \left(\frac{3 \pi}{4}, \pi \right),\\
\begin{pmatrix}
-\omega^2 \Ai (\omega^2 z) & \omega \Ai (\omega z)\\
\ii \omega \Ai'(\omega^2 z) & -\ii \omega^2 \Ai' (\omega z)
\end{pmatrix}, & \arg z \in \left(-\pi, -\frac{3 \pi}{4} \right),\\
\begin{pmatrix}
\Ai (z) &  \omega \Ai (\omega z)\\
-\ii \Ai'(z) & -\ii \omega^2 \Ai' (\omega z)
\end{pmatrix}, & \arg z \in \left(-\frac{3 \pi}{4}, 0\right).
\end{cases}
\end{align}

Furthermore, applying the asymptotics of $\Ai$ and $\Ai'$ in \cite[Chapter 9]{DLMF}, we have, as $z \to \infty$,
\begin{align}\label{asy:Ai}
\Phi^{({\Ai})}(z) \sim \frac{1}{2\sqrt{2}} \begin{pmatrix}
z^{-\frac 14} & 0\\
0 & z^{\frac 14}
\end{pmatrix} \begin{pmatrix}
1 & \ii\\
\ii & 1
\end{pmatrix}\begin{pmatrix}
\sum\limits_{k=0}^{\infty} \left(-\frac{3}{2}\right)^k \frac{\mathfrak{u}_k + \mathfrak{v}_k}{z^{3k/2}} & \ii \sum\limits_{k=0}^{\infty} \left(\frac{3}{2}\right)^k \frac{\mathfrak{u}_k - \mathfrak{v}_k}{z^{3k/2}}\\
-\ii \sum\limits_{k=0}^{\infty} \left(-\frac{3}{2}\right)^k \frac{\mathfrak{u}_k - \mathfrak{v}_k}{z^{3k/2}} & \sum\limits_{k=0}^{\infty} \left(\frac{3}{2}\right)^k \frac{\mathfrak{u}_k + \mathfrak{v}_k}{z^{3k/2}}
\end{pmatrix}e^{-\frac 23 z^{3/2} \sigma_3},
\end{align}
where $\mathfrak{u}_0=\mathfrak{v}_0=1$ and
\begin{align}\label{def:ukvk}
\mathfrak{u}_k = \frac{(6k-5)(6k-3)(6k-1)}{216(2k-1)k}\mathfrak{u}_{k-1}, \qquad \mathfrak{v}_k=\frac{6k+1}{1-6k}\mathfrak{u}_k.
\end{align}

\section{Asymptotic expansion of the Bessel functions for large order and large argument}\label{ap:B}
As a further application of our RH analysis performed in Section \ref{sec:rhp}, we establish transient asymptotics of the Bessel functions $J_{\nu}$ for large positive $\nu$ given in \cite[Equation (3.1)]{Olver52} and \cite[Formula 10.19.8]{DLMF}, which is described in the following lemma.
\begin{lemma}\label{lem:B}
As $\nu \to +\infty$, there holds
\begin{align}\label{eq:asyJnu}
J_{\nu}(\nu + \tau \nu^{\frac 13}) \sim \frac{2^{\frac 13}}{\nu^{\frac 13}}\Ai(-2^{\frac 13}\tau)\sum_{k=0}^{\infty}\frac{A_k(\tau)}{\nu^{\frac{2k}{3}}}+\frac{2^{\frac 23}}{\nu^{\frac 13}}\Ai'(-2^{\frac 13}\tau)\sum_{k=1}^{\infty}\frac{B_k(\tau)}{\nu^{\frac{2k}{3}}}
\end{align}
for any fixed complex number $\tau$, where $A_k(\tau)$ and $B_k(\tau)$ are certain polynomials of increasing degrees.
\end{lemma}
\begin{proof}
Recalling the definition of the matrix-valued function $\Psi_{\nu}$ in \eqref{def:psinu} and the following relations between $I_{\nu}$ and $J_{\nu}$ (cf. \cite[Formula 10.27.6]{DLMF})
\begin{align}
    I_{\nu}(z)=e^{\mp\frac{\nu\pi\ii}{2}}J_{\nu}\left(z e^{\pm\frac{\pi\ii}{2}}\right),\quad \pm \arg z \in \left(-\pi, \frac{\pi}{2}\right),
\end{align}
we have
\begin{equation}
J_{\nu}(\nu + \tau \nu^{\frac 13}) = \begin{pmatrix}
1 & 0
\end{pmatrix} \frac{1}{\sqrt{\pi}} e^{\frac{\pi \ii}{4} \pm \frac{\nu \pi \ii}{2}} \Psi_{\nu}(\nu^2 w) \begin{pmatrix}
1\\
0
\end{pmatrix}, \qquad \pm \Im w >0, 
\end{equation}
with $w = (1 + \tau \nu^{-\frac 23})^2$. Tracing back the transformations $\Psi \to Y \to T \to S$ defined in \eqref{def:Y}, \eqref{def:T} and \eqref{def:S}, we obtain
\begin{align}
\Psi_{\nu}(\nu^2 w)=\begin{pmatrix}
1 & 0\\
-\frac{4 \nu^2+3}{8} & 1
\end{pmatrix} \nu^{-\frac 12 \sigma_3}S(w)e^{-\nu g(w)\sigma_3}\begin{cases}
\begin{pmatrix}
1 & 0\\
-e^{-\pi \ii \nu} & 1
\end{pmatrix}, \quad & \arg(w-1) \in \left(0, \frac{\pi}{3}\right),\\
I, \quad & \arg(w-1) \in \left(\frac{\pi}{3},\frac{5\pi}{3}\right),\\
\begin{pmatrix}
1 & 0\\
-e^{-\pi \ii \nu} & 1
\end{pmatrix}, \quad & \arg(w-1) \in \left(\frac{2\pi}{3},2\pi\right).
\end{cases}
\end{align}
For any fixed complex number $\tau$, it is evident that $w= (1 + \tau \nu^{-\frac 23})^2 \in D(1,\varepsilon)$ for large $\nu$. This, together with \eqref{def:R} and \eqref{def:P1}, implies that 
\begin{equation}\label{exp-bessel}
J_{\nu}(\nu + \tau \nu^{\frac 13}) = 2^{\frac 12}\nu^{-\frac 12}e^{\frac{\pi \ii}{4}} \begin{pmatrix}
1 & 0
\end{pmatrix} R((1 + \tau \nu^{-\frac 23})^2)E( (1 + \tau \nu^{-\frac 23})^2)\begin{pmatrix}
\Ai(\nu^{\frac 23}f( (1 + \tau \nu^{-\frac 23})^2)) \\-\ii \Ai'(\nu^{\frac 23}f( (1 + \tau \nu^{-\frac 23})^2))
\end{pmatrix},
\end{equation}
where the functions $E$ and $R$ are given in \eqref{def:E} and \eqref{def:R}, respectively. Similar to the expansions \eqref{93}, \eqref{94} and \eqref{110}, we have from \eqref{E1}, \eqref{def:R1}, \eqref{R1} and \eqref{def:f} that, as $\nu \to +\infty$,
\begin{align}\label{exp-E}
E((1 + \tau \nu^{-\frac 23})^2)&= \nu^{\frac 16 \sigma_3}\left(I - \frac{\tau}{5}\sigma_3\nu^{-\frac 23}+ \begin{pmatrix}
\frac{3}{35}\tau^2 & 0\\
0 & -\frac{8}{175}\tau^2
\end{pmatrix}\nu^{-\frac 43} + \Boh(\nu^{-2}) \right) \nonumber \\
&\quad \times 2^{-\frac 16 \sigma_3}e^{-\frac{\pi \ii}{4} \sigma_3} \sigma_3,\\
R((1 + \tau \nu^{-\frac 23})^2)&=2^{\frac{1}{12} \sigma_3}\nu^{\frac 16 \sigma_3} \left(I + \begin{pmatrix}
0 & 0\\
\frac{7\cdot 2^{\frac 16}}{40} &0
\end{pmatrix}\nu^{-\frac 23}+\begin{pmatrix}
0&\frac{2^{-\frac 16}}{70} \nonumber \\
-\frac{2^{\frac 16}}{25}\tau&0
\end{pmatrix}\nu^{-\frac 43} + \Boh(\nu^{-2})\right)\label{exp-R}\\
&\quad \times 2^{-\frac{1}{12} \sigma_3}\nu^{-\frac 16 \sigma_3},
\end{align}
and
\begin{align}
\nu^{\frac 23}f((1 + \tau \nu^{-\frac 23})^2) &= -2^{\frac 13}\tau + \sum_{j=1}^{\infty} \hat p_{1,j}(\tau) \nu^{-\frac{2j}{3}},\label{exp-f}
\end{align}
where $\hat p_{1,j}(\tau)$ are certain polynomials of degree $j+1$ with
\begin{align}
\hat p_{1,1}(\tau) = \frac{3\cdot 2^{\frac 13}}{10} \tau^2,\qquad
\hat p_{1,2}(\tau)  = -\frac{32\cdot 2^{\frac 13}}{175} \tau^3.
\end{align}
Substituting \eqref{exp-E}--\eqref{exp-f} into \eqref{exp-bessel}, we obtain  \eqref{eq:asyJnu} after a simplification, 
where $A_k(\tau)$ and $B_k(\tau)$ are certain polynomials in $\tau$ with the first few being
\begin{align}
A_0(\tau) = 1, \quad A_1(\tau) =-\frac{1}{5}\tau, \quad B_1(\tau) = \frac{3}{10}\tau^2.
\end{align}
This completes the proof of Lemma \ref{lem:B}. 
\end{proof}
\begin{remark}
The above proof is actually valid for $|(1 + \tau \nu^{-\frac 23})^2-1|<\varepsilon$ with $\varepsilon \in (0,1)$ being any fixed small constant. Since 
\begin{align}
|(1 + \tau \nu^{-\frac 23})^2-1|=|2\tau \nu^{-\frac 23} + \tau^2 \nu^{-\frac 43}|<2|\tau| \nu^{-\frac 23} + |\tau|^2 \nu^{-\frac 43},
\end{align}
it is readily seen that the expansion \eqref{eq:asyJnu} is also valid provided 
\begin{equation}
|\tau| < (\sqrt{1+\varepsilon}-1)\nu^{\frac 23}.
\end{equation}

\end{remark}

\end{appendix}
%
%

\section*{Acknowledgements}
We thank Folkmar Bornemann for helpful comments on the preliminary version of this paper and the anonymous referees for their careful reading and constructive suggestions. Luming Yao was partially supported by National Natural Science Foundation of China under grant number 12401316 and Scientific Foundation for Youth Scholars of Shenzhen University under grant number 868--000001032818. Lun Zhang was partially supported by National Natural Science Foundation of China under grant numbers 12271105, 11822104, and ``Shuguang Program'' supported by Shanghai Education Development Foundation and Shanghai Municipal Education Commission.

\end{document}